\documentclass[11pt,a4paper,titlepage]{article}
\usepackage[utf8x]{inputenc}
\usepackage[T2A]{fontenc}
\usepackage[english]{babel}
\usepackage{mathrsfs}
\usepackage{ucs}
\usepackage{amsmath}
\usepackage{amsfonts}
\usepackage{amssymb}
\usepackage{amsthm}
\usepackage{longtable}
\usepackage{enumerate}
\usepackage{bussproofs}
\usepackage{proof}
\usepackage{xcolor}
\usepackage[width=0.00cm, height=0.00cm, left=2.00cm, right=2.00cm, top=2.00cm, bottom=2.00cm]{geometry}
\usepackage{url}

\theoremstyle{definition}

\newtheorem{theorem}{\textbf{Theorem}}
\newtheorem{lemma}{\textbf{Lemma}}
\newtheorem{proposition}{ \textbf{Proposition}}

\newtheorem{claim}{Claim}[section]

\newcommand{\1}{\vspace{.1in}}

\newcommand{\bc}{\begin{center}}
	\newcommand{\ec}{\end{center}}
\usepackage{bussproofs,longtable}
\usepackage[breaklinks,colorlinks,citecolor=blue,urlcolor=blue]{hyperref}
\newcommand\DOI[2]{DOI: \href{#1#2}{#2}}
\newcommand{\ts}{\,}

\newenvironment{Published}{\begin{list}{}{\leftmargin=25pt}
    \item[]\small\noindent{Published as}}%
    {\end{list}}
    \newenvironment{Funding}{\begin{list}{}{\leftmargin=25pt}
        \item[]\small\noindent{The research}}%
        {\end{list}}
\newenvironment{Abstract}{\begin{list}{}{\leftmargin=25pt}
    \item[]\footnotesize\noindent{\sc \bf Abstract:}}%
{\end{list}}
\newenvironment{keywords}{\begin{list}{}{\leftmargin=25pt}
    \item[]\footnotesize\noindent{\sc \bf Keywords:}}%
{\end{list}}

\begin{document}
\begin{center}
\vspace{0.3cm}
{\large Andrzej Indrzejczak, 
Yaroslav Petrukhin\vspace{0.3cm}\\
\textbf{Uniform Cut-free Bisequent Calculi for Three-valued Logics}}

\ \\
\begin{footnotesize}
\noindent Department of Logic, University of Lodz, Poland,  \url{andrzej.indrzejczak@filhist.uni.lodz.pl},\\
\noindent Center for Philosophy of Nature, University of Lodz, Poland,  \url{yaroslav.petrukhin@gmail.com}\\
\end{footnotesize}
\end{center}
\begin{Published}
 Indrzejczak, A., Petrukhin, Y. (2024)
Uniform Cut-Free Bisequent Calculi for Three-Valued Logics. Logic and Logical Philosophy. Vol. 33, no. 3, pp. 463--506. 
\DOI{http://dx.doi.org/}{10.12775/LLP.2024.019}
\bigskip
\end{Published}
\begin{Funding}
in this paper was funded by the European Union
 (ERC, ExtenDD, project number: 101054714). Views and opinions expressed
 are however those of the author(s) only and do not necessarily reflect those of
 the European Union or the European Research Council. Neither the European
 Union nor the granting authority can be held responsible for them.
\end{Funding}
 \begin{Abstract}
We present a uniform characterisation of three-valued logics by means of the bisequent calculus (BSC). It is a generalised form of a sequent calculus (SC) where rules operate on the ordered pairs of ordinary sequents. BSC may be treated as the weakest kind of system in the rich family of generalised SC operating on items being some collections of ordinary sequents, like hypersequent and nested sequent calculi. It seems that for many non-classical logics, including some many-valued, paraconsistent and modal logics, the reasonably modest generalisation of standard SC offered by BSC is sufficient. In this paper we examine a variety of three-valued logics and show how they can be formalised in the framework of BSC. We present a constructive syntactic proof provided that these systems are cut-free, satisfy the subformula property, and allow one to prove the interpolation theorem in many cases.
 \end{Abstract}
 \begin{keywords}
Bisequent Calculus, Cut elimination, Many-valued Logic, Three-valued logic, Interpolation Theorem
 \end{keywords}
\paragraph{Preliminary note.} This paper is an extended version of the conference paper \cite{IndrzejczakPetrukhinCADE}. Here, we present the results in more detail, providing several examples of proofs in sections \ref{StrongKleene} and \ref{OtherLogics} which were missing in the conference paper. Moreover, two new sections have been added: section \ref{Comments} explains the nature of uniformity of BSC, section \ref{Syntactic} presents a constructive syntactical proof of admissibility of structural rules, including cut. Sections \ref{Intro}, \ref{Logics}, \ref{Interpolation} and \ref{Conclusion} are more or less the same as in \cite{IndrzejczakPetrukhinCADE}.

\section{Introduction}\label{Intro}
The aim of this paper is to provide a uniform characterisation of a variety of three-valued logics by means of a  simple cut-free generalised sequent calculus (SC) called the bisequent calculus (BSC).  It is the weakest kind of system in the rich family of generalised sequent calculi operating
on collections of ordinary sequents \cite{Indrzejczak1}. If we restrict our interest to structures built of two sequents only, we obtain a limiting case of either hypersequent or nested sequent calculi (see \cite{LellmannPoggiolesi,Poggiolesi} for more information about nested sequent calculi); it is what we call a \textit{bisequent calculus}. 

 Is such a restricted calculus of any use? Hypersequent calculi already may be seen as quite restrictive form of generalised SC, yet they have been shown to be useful in many fields (see, e.g., \cite{Indrzejczak2} for a survey of applications of hypersequent calculi in modal logic, and \cite{MetcalfeOlivettiGabbay} for their use in fuzzy logic). BSC is even more restrictive but preliminary work on its application is promising.
It have already been succesfully applied by Indrzejczak to first-order modal logic {\bf S5} \cite{Indrzejczak1} and to the class of four-valued quasi-relevant logics \cite{Indrzejczak3}\footnote{A similar sequent system, but with a different semantic interpretation, was recently applied also by Fjellstad \cite{Fjellstad1,Fjellstad2}.}. In what follows we will focus on another application of such a minimal framework -- to three-valued logics.

Several proof systems of different kinds have been proposed so far for many-valued logics (see e.g. \cite{hah:tab01} or \cite{baaz01} for a survey).
The most direct and popular approach to the construction of many-valued sequent or tableau systems is based on the idea of the syntactic
representation of $n$ values either by means of $n$-sided sequents or by $n$ labels attached to formulae or sets of formulae. This solution has been
presented by many authors and despite its popularity has many drawbacks (see \cite{Indrzejczak2}  for discussion). A significant improvement in the construction of efficient SC or tableau systems for many-valued logic was proposed independently by Doherty \cite{doh:con91} and H\"ahnle \cite{hah:aut94}, where labels
correspond not to single values but to their sets (set-as-signs approach). Although BSC is a strictly syntactical calculus its semantic interpretation makes it more similar to this latter approach; a fuller discussion of this issue is provided in \cite{Indrzejczak3}.

BSC is uniform in the sense that all three-valued logics are characterised by the same set of axiomatic sequents, and in the case of logics having the same set of connectives (i.e. defined in the same way) the rules are identical even if the set of designated values or the consequence relation is defined in a different way. In this sense BSC is more uniform than several other approaches where either the set of axioms must be changed or rules for connectives must be different (even if described by means of the same table). In particular, BSC is superior in this respect to the generalised calculus for three-valued logics presented in \cite{Indrzejczak2}.

Section \ref{Logics} has an encyclopaedic character and is devoted to a description of a representative selection of three-valued logics. Section \ref{StrongKleene} contains a case study of BSC for {\bf K$_3$} and {\bf LP}. In Section \ref{OtherLogics}, we provide rules for connectives of all logics introduced in section \ref{Logics}. Section \ref{Comments} illustrates the process of encoding arbitrary operations as rules of BSC. Section \ref{Syntactic} is devoted to the syntactic proof of admissibility of structural rules. Section \ref{Interpolation} shows how BSC can be applied to prove interpolation for some three-valued paraconsistent and paracomplete logics. Section \ref{Conclusion} makes concluding remarks.

\section{Logics}\label{Logics}
We will examine several three-valued propositional logics determined by three element matrices with classical-like connectives (negation, disjunction, conjunction, and implication, plus the usual three-valued modal-style connectives). We are not going to consider other types of connectives except Pa\l{}asi\'{n}ska's connectives considered in section \ref{Comments}. 
The languages of these logics are freely generated algebras similar to three element algebras of values. Logics are interpreted by homomorphisms from languages to algebras such that $h(c^n(\varphi_1, \ldots, \varphi_n))= \underline{c}(h(\varphi_1), \ldots, h(\varphi_n))$ for every $n-$argument connective $c$ and the corresponding operation $\underline{c}$.  

Let us consider as the starting point two three element Kleene's algebras of the form: $\mathfrak{A}_3 = \langle A, O\rangle$ where $A=\{0, u, 1\}$ and $O$ contains an unary
operation $\underline{\neg}:A\longrightarrow A$ and binary operations $\odot: A\times A\longrightarrow A $,
where $\odot \in \{\underline{\wedge} , \underline{\vee} , \underline{\rightarrow} \}$. 
The operations are defined by
the following truth tables in the strong and weak Kleene algebra; the latter is also examined by Bochvar \cite{Bochvar} (negation is the same in both):

\begin{center}
	\begin{tabular}{|c|ccc|}
		\hline
		$\underline{\wedge} $ & 1 & u & 0 \\
		\hline
		1 & 1 & u & 0 \\
		u & u & u & 0 \\
		0 & 0 & 0 & 0 \\
		\hline
	\end{tabular}
	\begin{tabular}{|c|ccc|}
		\hline$\underline{\vee} $ & 1 & u & 0 \\\hline
		1 & 1 & 1 & 1 \\
		u & 1 & u & u \\
		0 & 1 & u & 0 \\\hline
	\end{tabular}
	\begin{tabular}{|c|ccc|}
		\hline
		$\underline{\rightarrow} $ & 1 & u & 0 \\
		\hline
		1 & 1 & u & 0 \\
		u & 1 & u & u \\
		0 & 1 & 1 & 1 \\
		\hline
	\end{tabular}
	\begin{tabular}{|c|c|}
		\hline  & $\underline{\neg} $\\\hline
		1 & 0\\
		u & u \\
		0 & 1 \\\hline
	\end{tabular}
\end{center} 
\begin{center} 	
\begin{tabular}{|c|ccc|}
	\hline
	$\underline{\wedge_w} $ & 1 & u & 0 \\
	\hline
	1 & 1 & u & 0 \\
	u & u & u & u \\
	0 & 0 & u & 0 \\
	\hline
\end{tabular}
\begin{tabular}{|c|ccc|}
	\hline$\underline{\vee_w} $ & 1 & u & 0 \\\hline
	1 & 1 & u & 1 \\
	u & u & u & u \\
	0 & 1 & u & 0 \\\hline
\end{tabular}
\begin{tabular}{|c|ccc|}
	\hline
	$\underline{\rightarrow_w} $ & 1 & u & 0 \\
	\hline
	1 & 1 & u & 0 \\
	u & u & u & u \\
	0 & 1 & u & 1 \\
	\hline
\end{tabular}
\end{center}

We obtain four matrices by specifying a set of designated values $D$ either as $\{ 1 \}$ or $\{1, u \}$. 
These are called $\mathfrak{SM}_1^3$, $\mathfrak{SM}_2^3$, $\mathfrak{WM}_1^3$ and $\mathfrak{WM}_2^3$ (where $\mathfrak{S}$ stands for strong, $\mathfrak{W}$ for weak, 1 and 2 indicate the amount of designated values). In general we will call matrices with $D = \{ 1 \}$ 1-matrices, and with $D= \{1, u \}$ -- 2-matrices. Accordingly we will also call logics determined by 1-matrices and 2-matrices, 1- and 2-logics respectively.
For any matrix we define a relation of matrix consequence in the following way:

\1

$\Gamma \models_M\varphi $ iff for any homomorphism $h$:
if $h(\Gamma)\subseteq D$, then $h(\varphi)\in D$. 

\1

Logics are identified with their matrix consequences. In particular,
logics determined by these matrices are {\bf K$_3$} (strong Kleene 1-logic) \cite{Kleene}, {\bf LP} -- the logic of paradox of Asenjo and Priest (corresponding 2-logic) \cite{Asenjo,Priest}, {\bf K}$_3^{\bf w}$ (weak Kleene 1-logic) \cite{Kleene}, {\bf PWK} (paraconsistent weak Kleene 2-logic) of Halld\'{e}n \cite{Hallden}. 

Let us consider a few modifications of strong and weak Kleene logics. Here is McCarthy's logic {\bf K$_3^\rightarrow$} \cite{McCarty} (also called Kleene's sequential and studied by Fitting \cite{Fitting94})  
and its interesting modification presented by Ko\-men\-dant\-skaya \cite{Komendantskaya} under the name  {\bf K$_3^\leftarrow$} by means of the following truth tables (again, negation is unchanged):

\begin{center}
	\begin{tabular}{|c|ccc|}
		\hline
		$\underline{\wedge_{mC}} $ & 1 & u & 0 \\
		\hline
		1 & 1 & u & 0 \\
		u & u & u & u \\
		0 & 0 & 0 & 0 \\
		\hline
	\end{tabular}
	\begin{tabular}{|c|ccc|}
		\hline$\underline{\vee_{mC}} $ & 1 & u & 0 \\\hline
		1 & 1 & 1 & 1 \\
		u & u & u & u \\
		0 & 1 & u & 0 \\\hline
	\end{tabular}
	\begin{tabular}{|c|ccc|}
		\hline
		$\underline{\rightarrow_{mC}} $ & 1 & u & 0 \\
		\hline
		1 & 1 & u & 0 \\
		u & u & u & u \\
		0 & 1 & 1 & 1 \\
		\hline
	\end{tabular}
	\end{center} 
	\begin{center}
	\begin{tabular}{|c|ccc|}
		\hline
		$\underline{\wedge_K} $ & 1 & u & 0 \\
		\hline
		1 & 1 & u & 0 \\
		u & u & u & 0 \\
		0 & 0 & u & 0 \\
		\hline
	\end{tabular}
	\begin{tabular}{|c|ccc|}
		\hline$\underline{\vee_K} $ & 1 & u & 0 \\\hline
		1 & 1 & u & 1 \\
		u & 1 & u & u \\
		0 & 1 & u & 0 \\\hline
	\end{tabular}
	\begin{tabular}{|c|ccc|}
		\hline
		$\underline{\rightarrow_K} $ & 1 & u & 0 \\
		\hline
		1 & 1 & u & 0 \\
		u & 1 & u & u \\
		0 & 1 & u & 1 \\
		\hline
	\end{tabular}
\end{center} 

Both {\bf K$_3^\rightarrow$} and {\bf K$_3^\leftarrow$} are logics determined by 1-matrices. An important property of {\bf K$_3$}, {\bf K}$_3^{\bf w}$, {\bf K$_3^\rightarrow$}, and {\bf K$_3^\leftarrow$} is that they are the only three-valued logics with one designated value which produce partial recursive predicates (see \cite{Kleene,Komendantskaya} for more details).

Several other important logics are obtained by changing the definitions of $\rightarrow$ and $ \neg $. Consider \L ukasiewicz's \cite{Lukasiewicz}, S\l upecki's \cite{Slupecki67}, Heyting's \cite{Heyting} implications as well as Heyting's \cite{Heyting}, Bochvar's \cite{Bochvar}, Post's and dual Post's \cite{Post21} negations. Let us also consider yet another pair of additive conjunction and disjunction, arising in \L ukasiewicz's logic:

\begin{center}
	\begin{tabular}{|c|ccc|}
		\hline
		$\underline{\rightarrow_L} $ & 1 & u & 0 \\
		\hline
		1 & 1 & u & 0 \\
		u & 1 & 1 & u \\
		0 & 1 & 1 & 1 \\
		\hline
	\end{tabular}
	\begin{tabular}{|c|ccc|}
		\hline $\underline{\rightarrow_{Sl}} $ & 1 & u & 0\\\hline
		1 & 1 & u & 0 \\
		u & 1 & 1 & 1 \\
		0 & 1 & 1 & 1 \\	\hline
	\end{tabular}
	\begin{tabular}{|c|ccc|}
		\hline $\underline{\rightarrow_H} $ & 1 & u & 0\\\hline
		1 & 1 & u & 0 \\
		u & 1 & 1 & 0 \\
		0 & 1 & 1 & 1 \\	\hline
	\end{tabular}
	\begin{tabular}{|c|c|}
		\hline  & $\underline{\neg_H} $\\\hline
		1 & 0\\
		u & 0 \\
		0 & 1 \\\hline
	\end{tabular}
	\end{center} 
	\begin{center}
	\begin{tabular}{|c|c|}
		\hline  & $\underline{\neg_B} $\\\hline
		1 & 0\\
		u & 1 \\
		0 & 1 \\\hline
	\end{tabular}
	\begin{tabular}{|c|c|}
		\hline  & $\underline{\neg_P} $\\\hline
		1 & u\\
		u & 0 \\
		0 & 1 \\\hline
	\end{tabular}
	\begin{tabular}{|c|c|}
		\hline  & $\underline{\neg_{\mathit{DP}}} $\\\hline
		1 & 0\\
		u & 1 \\
		0 & u \\\hline
	\end{tabular}
	\begin{tabular}{|c|ccc|}
		\hline
		$\underline{\wedge_L} $ & 1 & u & 0 \\
		\hline
		1 & 1 & u & 0 \\
		u & u & 0 & 0 \\
		0 & 0 & 0 & 0 \\
		\hline
	\end{tabular}
	\begin{tabular}{|c|ccc|}
		\hline
		$\underline{\vee_L} $ & 1 & u & 0 \\
		\hline
		1 & 1 & 1 & 1 \\
		u & 1 & 1 & u \\
		0 & 1 & u & 0 \\
		\hline
	\end{tabular}
\end{center}

$\mathfrak{SM}_1^3$ with \L ukasiewicz's implication (instead of Kleene's one) yields \L ukasiewicz's famous {\bf \L$_3$}, the first many-valued logic. In {\bf \L$_3$} we may deal with two pairs of conjunction and disjunction.  
We have: $ \varphi\vee\psi=(\varphi\rightarrow_L\psi)\rightarrow_L\psi $ and $\varphi\wedge\psi=\neg(\neg\varphi\vee\neg\psi)  $, but $  \varphi\wedge_L\psi=\neg(\varphi\rightarrow_L\neg \psi)$ and $ \varphi\vee_L\psi=\neg \varphi\rightarrow_L \psi$. $\mathfrak{SM}_1^3$ with S\l upecki's implication is an alternative to {\bf \L$_3$} having the deduction theorem. It was studied by S\l{}upecki, Bryll, and Prucnal \cite{Slupecki67} as well as Avron  \cite{Avron}, under the name {\bf GM$_3$}. If we change negation and implication of $\mathfrak{SM}_1^3$ to Heyting's ones, then we get Heyting's \cite{Heyting} logic {\bf G$_3$}, a close relative of intuitionistic logic (named after G\"{o}del who also studied this logic \cite{Godel}; this logic was investigated by Ja\'{s}kowski as well \cite{Jaskowski}). 
The disjunction of $\mathfrak{SM}_1^3$ and Post's cyclic negation form Post's logic {\bf P$_3$} \cite{Post21} which is known for being functionally complete in the three-valued setting. In \cite{PetrukhinP}, a dual cyclic negation $ \neg_{\mathit{DP}} $ was suggested (it reverses the direction of cyclicality of Post's negation). $\mathfrak{SM}_2^3$ with Heyting's implication and Bochvar's negation was investigated by Osorio and Carballido \cite{OsorioCarballido08} under the name {\bf G}$_3^\prime$. 

In the case of $\mathfrak{SM}_2^3$ the following connectives are interesting as well: Soboci\'nski's \cite{Sobocinski} conjunction, disjunction and two implications, as well as D'Ottaviano/DaCosta/Ja\'{s}kowski/S\l{}upecki's \cite{DOttavianodaCosta,Jaskowski48,Slupecki39} implication $\rightarrow_J$:

\begin{center}
	\begin{tabular}{|c|ccc|}
		\hline
		$\underline{\wedge_S} $ & 1 & u & 0 \\
		\hline
		1 & 1 & 1 & 0 \\
		u & 1 & u & 0 \\
		0 & 0 & 0 & 0 \\
		\hline
	\end{tabular}
	\begin{tabular}{|c|ccc|}
		\hline$\underline{\vee_S} $ & 1 & u & 0 \\\hline
		1 & 1 & 1 & 1 \\
		u & 1 & u & 0 \\
		0 & 1 & 0 & 0 \\\hline
	\end{tabular}
	\begin{tabular}{|c|ccc|}
		\hline
		$\underline{\rightarrow_S} $ & 1 & u & 0 \\
		\hline
		1 & 1 & 0 & 0 \\
		u & 1 & u & 0 \\
		0 & 1 & 1 & 1 \\
		\hline
	\end{tabular}
	\end{center} 
	\begin{center}
	\begin{tabular}{|c|ccc|}
		\hline
		$\underline{\rightarrow_S^\prime} $ & 1 & u & 0 \\
		\hline
		1 & 1 & u & 0 \\
		u & 1 & 1 & 0 \\
		0 & 1 & u & 1 \\
		\hline
	\end{tabular}
	\begin{tabular}{|c|ccc|}
		\hline$\underline{\rightarrow_J} $ & 1 & u & 0 \\\hline
		1 & 1 & u & 0 \\
		u & 1 & u & 0 \\
		0 & 1 & 1 & 1 \\\hline
	\end{tabular}
\end{center}

Soboci\'nski's logic {\bf S$_3$} is obtained from $\mathfrak{SM}_2^3$ by the replacement of all binary connectives of this matrix with Soboci\'nski's original ones. This logic may be treated as a relevant logic. However, a more popular three-valued relevant logic is Anderson and Belnap's {\bf RM$_3$} \cite{EntailmentI} which is obtained from $\mathfrak{SM}_2^3$ only by the replacement of its implication with Soboci\'nski's one. Note that earlier Soboci\'nski \cite{Sobocinski1} considered yet another implication $ \rightarrow_S^\prime $ defined above.

$\mathfrak{SM}_2^3$ with D'Ottaviano/DaCosta/Ja\'{s}kowski/S\l{}upecki's implication (first mentioned by S\l{}upecki \cite{Slupecki39}) instead of Kleene's one was independently studied by several authors: D'Ottaviano and da Costa themselves \cite{DOttavianodaCosta,daCosta}, Asenjo and Tamburino \cite{AsenjoTamburino}, Batens \cite{Batens} (under the name PI$^s$), Avron \cite{Avron86} (under the name {\bf RM$_3^\supset$}), and Rozonoer \cite{Rozonoer} (under the name \textbf{PCont}). An important extension of this logic is {\bf J$_3$} by D'Ottaviano and da Costa \cite{DOttavianodaCosta}. It has an additional connective which is \L ukasiewicz's tabular possibility operator (see below; we also present \L ukasiewicz's tabular necessity operator). 

\begin{center}
	\begin{tabular}{|c|ccc|}
		\hline
		$\underline{\wedge_C} $ & 1 & u & 0 \\
		\hline
		1 & 1 & 0 & 0 \\
		u & 0 & 0 & 0 \\
		0 & 0 & 0 & 0 \\
		\hline
	\end{tabular}
	\begin{tabular}{|c|ccc|}
		\hline$\underline{\vee_C} $ & 1 & u & 0 \\\hline
		1 & 1 & 1 & 1 \\
		u & 1 & 0 & 0 \\
		0 & 1 & 0 & 0 \\\hline
	\end{tabular}
	\begin{tabular}{|c|ccc|}
		\hline
		$\underline{\rightarrow_C} $ & 1 & u & 0 \\
		\hline
		1 & 1 & 0 & 0 \\
		u & 1 & 1 & 1 \\
		0 & 1 & 1 & 1 \\
		\hline
	\end{tabular}
\begin{tabular}{|c|c|}
	\hline & $\underline{\Diamond}$\\\hline
	1 & 1\\
	u & 1\\
	0 & 0\\\hline
\end{tabular}
\begin{tabular}{|c|c|}
	\hline & $\underline{\Box}$\\\hline
	1 & 1\\
	u & 0\\
	0 & 0\\\hline
\end{tabular}	
\end{center}

As is easy to guess, since {\bf RM$_3$} may be viewed as a relevant logic, it should be paraconsistent as well. Moreover, {\bf J$_3$}, {\bf S$_3$}, {\bf LP}, and many other three-valued logics with two designated values are paraconsistent (in contrast, three-valued logics with one designated value are paracomplete). One of the most famous three-valued paraconsistent logics is Sette's logic {\bf P$ ^1 $} \cite{Sette}. It has Bochvar's negation and the above presented binary connectives $\wedge_C, \vee_C$  (both 1 and u are designated). There is a version of {\bf P$ ^1 $} with Kleene's negation introduced by Carnielli and Marcos \cite{CarnielliMarcos,Marcos} and called {\bf P$ ^2 $}.

A paracomplete companion of {\bf P$ ^1 $}, the logic {\bf I$ ^1 $}, was presented by Sette and Carnielli \cite{SetteCarnieli}: it has Heyting's negation and the binary connectives presented below (the implication was first introduced by Bochvar \cite{Bochvar}). Its version with Kleene's negation is {\bf I$ ^2 $} due to Marcos \cite{Marcos}. Both {\bf I$ ^1 $} and {\bf I$ ^2 $} have one designated value.

	\begin{center}
		\begin{tabular}{|c|ccc|}
			\hline
			$\underline{\wedge_{Se}} $ & 1 & u & 0 \\
			\hline
			1 & 1 & 1 & 0 \\
			u & 1 & 1 & 0 \\
			0 & 0 & 0 & 0 \\
			\hline
		\end{tabular}
		\begin{tabular}{|c|ccc|}
			\hline$\underline{\vee_{Se}} $ & 1 & u & 0 \\\hline
			1 & 1 & 1 & 1 \\
			u & 1 & 1 & 1 \\
			0 & 1 & 1 & 0 \\\hline
		\end{tabular}
		\begin{tabular}{|c|ccc|}
			\hline
			$\underline{\rightarrow_{Se}} $ & 1 & u & 0 \\
			\hline
			1 & 1 & 1 & 0 \\
			u & 1 & 1 & 0 \\
			0 & 1 & 1 & 1 \\
			\hline
		\end{tabular}
\end{center} 
\begin{center}		
	\begin{tabular}{|c|ccc|}
		\hline
		$\underline{\rightarrow_R} $ & 1 & u & 0 \\
		\hline
		1 & 1 & 0 & 0 \\
		u & 1 & 1 & 0 \\
		0 & 1 & 1 & 1 \\
		\hline
	\end{tabular}
	\begin{tabular}{|c|ccc|}
		\hline
		$\underline{\rightarrow_T} $ & 1 & u & 0 \\
		\hline
		1 & 1 & 0 & 0 \\
		u & 1 & 1 & u \\
		0 & 1 & 1 & 1 \\
		\hline
	\end{tabular}
	
\end{center}

Last but not least, let us mention Rescher's \cite{Rescher} and Tomova's \cite{Tomova12} implications (given above). These implications can be added to $\mathfrak{SM}_1^3$. Tomova \cite{Tomova12} introduced the concept of natural implication. In the three-valued case with one designated value there are only 6 natural implications: \L ukasiewicz's, S\l{}upecki's, Heyting's, Bochvar's, Rescher's, and Tomova's. In the case with two designated values there are 24 natural implications, including Heyting's and Rescher's as well as both Soboci\'{n}ski's implications, D'Ottaviano/DaCosta/Ja\'{s}kowski's/S\l{}upecki's, and Sette's implications.

\section{Bisequent Calculus for K$_3$ (and LP)}\label{StrongKleene}

Bisequents in BSC are ordered pairs of sequents $\Gamma\Rightarrow\Delta \mid \Pi\Rightarrow\Sigma$, where $\Gamma, \Delta, \Pi, \Sigma$ are
finite (possibly empty) multisets of formulae. We will call the elements of a bisequent a 1- and 2-sequent respectively. Bisequents with all elements being atomic will be also called atomic. In what follows $ B $ stands for
arbitrary bisequents and $ S $ for sequents.

Let us define the calculus BSC-K$_3$ which provides an adequate formalization of {\bf K$_3$}. 
A bisequent $\Gamma\Rightarrow\Delta \mid \Pi\Rightarrow\Sigma$ is axiomatic iff it has nonempty $\Gamma\cap\Sigma$ or $\Gamma\cap\Delta$ or $\Pi\cap\Sigma$. 
In fact this set of axioms is fixed for all calculi considered in the next section. Also if constants $\top, \bot, U$ (the last for fixed undefined proposition) are added we must add axioms of the form: $\Gamma\Rightarrow\Delta, \top \mid \Pi\Rightarrow\Sigma$; $\Gamma\Rightarrow\Delta \mid \Pi\Rightarrow\Sigma, \top$; $\bot, \Gamma\Rightarrow\Delta \mid \Pi\Rightarrow\Sigma$; $\Gamma\Rightarrow\Delta \mid \bot, \Pi\Rightarrow\Sigma$; $U, \Gamma\Rightarrow\Delta \mid \Pi\Rightarrow\Sigma$ and $\Gamma\Rightarrow\Delta \mid \Pi\Rightarrow\Sigma, U$.

The set of rules characterising the operations of the strong Kleene algebra consists of the following schemata:

\1

\noindent\begin{tabular}{ll}

	$(\mathord{\neg}\mathord{\Rightarrow\mid})$ \ $\dfrac{\Gamma
		\Rightarrow \Delta \mid \Pi\Rightarrow\Sigma, \varphi}{\neg\varphi, \Gamma \Rightarrow \Delta \mid \Pi\Rightarrow\Sigma}$ &
	$(\mathord{\Rightarrow}\mathord{\neg\mid})$ \ $\dfrac{\Gamma
		\Rightarrow \Delta \mid \varphi, \Pi\Rightarrow\Sigma}{\Gamma
		\Rightarrow \Delta, \neg\varphi \mid \Pi\Rightarrow\Sigma}$\\[16pt]
	
	$(\mathord{\mid\neg}\mathord{\Rightarrow})$ \ $\dfrac{\Gamma
		\Rightarrow \Delta, \varphi \mid \Pi\Rightarrow\Sigma}{\Gamma \Rightarrow \Delta \mid \neg\varphi, \Pi\Rightarrow\Sigma}$ &
	$(\mathord{\mid\Rightarrow}\mathord{\neg})$ \ $\dfrac{\varphi, \Gamma
		\Rightarrow \Delta \mid \Pi\Rightarrow\Sigma}{\Gamma
		\Rightarrow \Delta \mid \Pi\Rightarrow\Sigma, \neg\varphi}$\\[16pt]
\end{tabular}
	
\noindent\begin{tabular}{ll}	
	$(\mathord{\wedge}\mathord{\Rightarrow\mid})$ \ $\dfrac{\varphi, \psi,
		\Gamma \Rightarrow \Delta \mid S}{\varphi\wedge\psi, \Gamma
		\Rightarrow \Delta \mid S}$ &
	 $(\mathord{\Rightarrow}\mathord{\wedge\mid})$ \ $\dfrac{\Gamma
		\Rightarrow \Delta, \varphi \mid S \qquad
		\Gamma\Rightarrow \Delta,
		\psi \mid S}{\Gamma \Rightarrow \Delta,
		\varphi\wedge\psi \mid S}$\\[16pt]
	
 $(\mathord{\mid\wedge}\mathord{\Rightarrow})$ \ $\dfrac{S \mid \varphi, \psi,
		\Gamma \Rightarrow \Delta}{S \mid \varphi\wedge\psi, \Gamma
		\Rightarrow \Delta}$ &
 $(\mathord{\mid\Rightarrow}\mathord{\wedge})$ \ $\dfrac{S\mid \Gamma
		\Rightarrow \Delta, \varphi \qquad
		S\mid \Gamma\Rightarrow \Delta,
		\psi}{S\mid\Gamma \Rightarrow \Delta,
		\varphi\wedge\psi }$\\[16pt]
\end{tabular}

\noindent\begin{tabular}{ll}	
	
 $(\mathord{\Rightarrow}\mathord{\vee}\mid)$ \ $\dfrac{\Gamma
		\Rightarrow \Delta, \varphi,\psi \mid S} {\Gamma\Rightarrow \Delta,
		\varphi\vee\psi \mid S}$ &
	 $(\mathord{\vee}\mathord{\Rightarrow}\mid)$ \ $\dfrac{\varphi, \Gamma
		\Rightarrow \Delta \mid S \qquad \psi, \Gamma
		\Rightarrow \Delta \mid S}{\varphi\vee\psi,
		\Gamma\Rightarrow \Delta \mid S}$\\[16pt]

	 $(\mathord{\mid\Rightarrow}\mathord{\vee})$ \ $\dfrac{S\mid\Gamma
		\Rightarrow \Delta, \varphi,\psi } {S\mid\Gamma\Rightarrow \Delta,
		\varphi\vee\psi}$ &
 $(\mathord{\mid\vee}\mathord{\Rightarrow})$ \ $\dfrac{S\mid\varphi, \Gamma
		\Rightarrow \Delta \qquad S\mid \psi, \Gamma
		\Rightarrow \Delta }{S\mid\varphi\vee\psi,
		\Gamma \Rightarrow \Delta}$\\[16pt]	
\end{tabular}
\begin{flushleft}
	$(\mathord{\Rightarrow}\mathord{\rightarrow}\mid)$ \, $\dfrac{\Gamma
		\Rightarrow \Delta, \psi \mid \varphi, \Pi\Rightarrow\Sigma}{\Gamma\Rightarrow\Delta, \varphi\rightarrow\psi \mid \Pi\Rightarrow\Sigma}$\qquad
	$(\mathord{\mid\Rightarrow}\mathord{\rightarrow})$ \, $\dfrac{\varphi, \Gamma
		\Rightarrow \Delta \mid\Pi\Rightarrow\Sigma, \psi}{\Gamma\Rightarrow\Delta \mid \Pi\Rightarrow\Sigma, \varphi\rightarrow\psi}$	
\end{flushleft}
\begin{flushleft}
	$(\mathord{\rightarrow}\mathord{\Rightarrow}\mid)$ \, $\dfrac{\Gamma
		\Rightarrow \Delta \mid \Pi\Rightarrow\Sigma, \varphi \quad\psi, \Gamma\Rightarrow\Delta \mid
		\Pi \Rightarrow \Sigma } {\varphi\rightarrow\psi, \Gamma
		\Rightarrow \Delta \mid \Pi\Rightarrow\Sigma}$
\end{flushleft}
\begin{flushleft}
$(\mathord{\mid\rightarrow}\mathord{\Rightarrow})$ \, $\dfrac{\Gamma
	\Rightarrow \Delta, \varphi \mid \Pi\Rightarrow\Sigma \quad \Gamma\Rightarrow\Delta \mid \psi,
	\Pi \Rightarrow \Sigma } {\Gamma
	\Rightarrow \Delta \mid \varphi\rightarrow\psi, \Pi\Rightarrow\Sigma}$	
\end{flushleft}

Note that all rules satisfy the subformula property and other desirable properties of well-behaved SC, like explicitness, separation, symmetry (see \cite{Poggiolesi,Wansing, Indrzejczak2} for a discussion). In particular, they are context independent in the sense that no side conditions are constrained and as a consequence all rules are permutable. This feature will be of special importance for the proof of the Interpolation Theorem in section \ref{Interpolation}. One may easily observe that in case of rules for strong $\wedge, \vee$ we have just standard G3 rules but repeated in both components. Rules for negation and implication have different character since side and principal formula are in different sequents in all cases.

Bisequents as such do not directly correspond to standard consequence relations in suitable matrices. Hence 
before we define the notion of a proof in BSC-K$_3$ (or any other logic) it is better to start with a more general concept. A proof-search tree for a bisequent $B$ in BSC-L, where L is any logic, is a tree of bisequents with $B$ as the root and nodes generated by rules of BSC-L. A proof-search tree is complete iff every leaf is atomic, and it is axiomatic iff all leaves are axiomatic. The height of a proof-search tree is defined as the length of the maximal branches. A simple consequence of the subformula property of rules is:

\begin{proposition}
	Every proof-search tree may be extended to a complete proof-search tree
\end{proposition}

The notion of a proof in BSC-K$_3$ is introduced not only by restricting the class of proof-search trees in BSC-K$_3$ to axiomatic ones but also  by restricting the class of admissible roots. In general the rationale for bisequents is that a 1-sequent corresponds to a consequence relation in 1-matrices and a 2-sequent to a consequence relation in 2-matrices. Since {\bf K$_3$} is characterised by 1-matrix we have:

\1

BSC-K$_3  \vdash B$ iff there is an axiomatic proof-search tree for $B := \Gamma\Rightarrow\varphi \mid \Rightarrow$.

\1

We define the L-validity (L-satisfiability) of bisequents in the following way:

\1

{\bf L} $\models\Gamma\Rightarrow\Delta \mid \Pi\Rightarrow\Sigma$ iff 
every homomorphism $h$ satisfies $\Gamma\Rightarrow\Delta \mid \Pi\Rightarrow\Sigma$. The latter holds for $h$ iff for some $\varphi$: either $\varphi\in \Gamma$ and $h(\varphi)\neq 1$ or $\varphi\in \Delta$ and $h(\varphi)=1$ or $\varphi\in \Pi$ and $h(\varphi)=0$ or $\varphi\in \Sigma$ and $h(\varphi)\neq 0$.

\1

Clearly {\bf L} $\not \models\Gamma\Rightarrow\Delta \mid \Pi\Rightarrow\Sigma$ iff for some $h$, all elements of $\Gamma$ are true, all elements of $\Delta$ are either false or undefined, all elements of $\Pi$ are either true or undefined and all elements of $\Sigma$ are false. In this case we say that $h$ falsifies this sequent.

\1

Obviously, all axiomatic bisequents are valid for any logic L. As for the rules they are not only sound (i.e. validity-preserving) but also invertible; namely the following holds:

\begin{theorem}\label{Theorem1}
	For all rules of BSC-K$_3$, all premisses are {\bf K$_3$}-valid iff the conclusion is {\bf K$_3$}-valid.
\end{theorem}
\begin{proof}
Straightforward proof by tedious checking. For illustration we demonstrate the case of $(\mathord{\mid\rightarrow}\mathord{\Rightarrow})$. Suppose that ${\bf K_3}\models\Gamma\Rightarrow\Delta,\varphi \mid \Pi\Rightarrow\Sigma$ and ${\bf K_3}\models\Gamma\Rightarrow\Delta \mid \psi,\Pi\Rightarrow\Sigma$, while ${\bf K_3}\not\models\Gamma\Rightarrow\Delta \mid \varphi\rightarrow\psi,\Pi\Rightarrow\Sigma$. Then (1) for each $ h $ there is a $ \chi $ such that ($ \chi\in\Gamma $ and $ h(\chi)\not=1 $) or ($ \chi\in\Delta $ and $ h(\chi)=1 $) or ($ \chi\in\Pi $ and $ h(\chi)=0$) or ($ \chi\in\Sigma $ and $ h(\chi)\not=0 $) or ($ \chi=\varphi $ and $ h(\varphi)=1 $); (2) for each $ h $ there is an $ \omega $ such that ($ \omega\in\Gamma $ and $ h(\omega)\not=1 $) or ($ \omega\in\Delta $ and $ h(\omega)=1 $) or ($ \omega\in\Pi $ and $ h(\omega)=0$) or ($ \omega\in\Sigma $ and $ h(\omega)\not=0 $) or ($ \omega=\psi $ and $ h(\psi)=0 $); (3) for some $ h $, it holds that, for each $ \gamma\in\Gamma $, $h(\gamma)=1$, for each $ \delta\in\Delta $, $ h(\delta)\not=1 $, for each $ \pi\in\Pi $, $ h(\pi)\not=0 $, for each $ \sigma\in\Sigma $, $ h(\sigma)=0 $, and $ h(\varphi\rightarrow\psi)\not=0 $. From (1) and (3) we get (4) $ h(\varphi)=1 $. From (2) and (3) we get (5) $ h(\psi)=0 $. From (4) and (5) we obtain (6) $ h(\varphi\rightarrow\psi)=0 $. However, from (3) follows that $ h(\varphi\rightarrow\psi)\not=0 $. Contradiction. Hence, ${\bf K_3}\models\Gamma\Rightarrow\Delta \mid \varphi\rightarrow\psi,\Pi\Rightarrow\Sigma$.

Suppose that ${\bf K_3}\models\Gamma\Rightarrow\Delta \mid \varphi\rightarrow\psi,\Pi\Rightarrow\Sigma$, while ${\bf K_3}\not\models\Gamma\Rightarrow\Delta,\varphi \mid \Pi\Rightarrow\Sigma$ or ${\bf K_3}\not\models\Gamma\Rightarrow\Delta \mid \psi,\Pi\Rightarrow\Sigma$. Then (1) for each $ h $ there is a $ \chi $ such that ($ \chi\in\Gamma $ and $ h(\chi)\not=1 $) or ($ \chi\in\Delta $ and $ h(\chi)=1 $) or ($ \chi\in\Pi $ and $ h(\chi)=0$) or ($ \chi\in\Sigma $ and $ h(\chi)\not=0 $) or ($ \chi=\varphi\rightarrow\psi $ and $ h(\varphi\rightarrow\psi)=0 $); (2a) for some $ h $, it holds that, for each $ \gamma\in\Gamma $, $h(\gamma)=1$, for each $ \delta\in\Delta $, $ h(\delta)\not=1 $, for each $ \pi\in\Pi $, $ h(\pi)\not=0 $, for each $ \sigma\in\Sigma $, $ h(\sigma)=0 $, and $ h(\varphi)\not=1 $ or (2b) for some $ h $, it holds that, for each $ \gamma\in\Gamma $, $h(\gamma)=1$, for each $ \delta\in\Delta $, $ h(\delta)\not=1 $, for each $ \pi\in\Pi $, $ h(\pi)\not=0 $, for each $ \sigma\in\Sigma $, $ h(\sigma)=0 $, and $ h(\psi)\not=0 $. From (1) we obtain that $ h(\varphi\rightarrow\psi)=0 $, and hence (3a) $ h(\varphi)=1 $ and (3b) $ h(\psi)=0 $. Suppose that (2a) holds. Then (4a) $ h(\varphi)\not=1 $. Contradiction. Suppose that (2b) holds. Then (4b) $ h(\psi)\not=0 $. Contradiction. Therefore, ${\bf K_3}\models\Gamma\Rightarrow\Delta,\varphi \mid \Pi\Rightarrow\Sigma$ and ${\bf K_3}\models\Gamma\Rightarrow\Delta \mid \psi,\Pi\Rightarrow\Sigma$.
\end{proof}

A simple consequence of this theorem is that for every rule the conclusion is falsified by some $h$ iff at least one premiss is falsified by the same $h$.

\begin{theorem}[Soundness]
	If BSC-K$_3 \vdash \Gamma\Rightarrow\varphi\mid\,\Rightarrow$, then $\Gamma\models_{K_3}\varphi$
\end{theorem}
\begin{proof}
By induction on the height of the proof, use Theorem \ref{Theorem1}.
\end{proof}

The invertibility of all rules, which will be shown also syntactically in Section \ref{Syntactic}, implies that proof search process is confluent, i.e. that the order of applications of rules does not affect the result. In particular, $B$ is provable iff every proof-search tree may be extended to obtain a proof.

\begin{theorem}[Completeness]
	If $\Gamma\models_{K_3}\varphi$, then
	BSC-K$_3 \vdash \Gamma\Rightarrow\varphi\mid\,\Rightarrow$.
\end{theorem}
\begin{proof}
	Assume that $\Gamma\models_{K_3}\varphi$ but BSC-K$_3 \nvdash \Gamma\Rightarrow\varphi \mid\,\Rightarrow$. Hence in every complete proof-search tree for $\Gamma\Rightarrow\varphi \mid\,\Rightarrow$ there is at least one branch starting with nonaxiomatic atomic bisequent falsified by some $h$. Since all rules inherit this valuation, then the root is also falsified contrary to our assumption.
\end{proof}

As a simple consequence we obtain also a decision procedure for {\bf K$_3$} (and for other logics L with complete BSC-L).

Another by-product of our proof is that the following cut rules are admissible in BSC-K$_3$ (and other logics):

\begin{flushleft}\label{Cut}
	$(Cut\mid)$ \ $\dfrac{\Gamma 
	\Rightarrow  \Delta, \varphi \mid \Lambda\Rightarrow\Theta \qquad
	\varphi, \Pi  \Rightarrow  \Sigma \mid \Xi\Rightarrow\Omega} {\Gamma, \Pi  \Rightarrow  \Delta, \Sigma \mid \Lambda, \Xi\Rightarrow\Theta, \Omega}$
\end{flushleft}

\begin{flushleft}
	$(\mid Cut)$ \ $\dfrac{\Gamma 
	\Rightarrow  \Delta \mid \Lambda\Rightarrow\Theta, \varphi \qquad
	\Pi  \Rightarrow  \Sigma \mid \varphi, \Xi\Rightarrow\Omega} {\Gamma, \Pi  \Rightarrow  \Delta, \Sigma \mid \Lambda, \Xi\Rightarrow\Theta, \Omega}$
\end{flushleft}

Moreover, we can constructively prove that these cut rules are admissible. We do it in Section \ref{Syntactic}.

Note that the rules stated above provide BSC not only for {\bf K$_3$} but also for {\bf LP}. The only difference is that in {\bf LP} we consider as provable bisequents of the form $\Rightarrow \mid \Gamma\Rightarrow\varphi$ which is a consequence of the fact that it is determined by a 2-matrix.
All the results established for BSC-K$_3$ hold for BSC-LP.

\section{Bisequent Calculi for Other Logics}\label{OtherLogics}

We provide sets of rules adequate for all logics described in Sect. \ref{Logics}. Every operation
will be characterised by four rules of introduction to antecedents and consequents
of 1- and 2-sequents. In every case it holds that either

\1

$\Gamma\models_L \varphi$ iff BSC-L $\vdash \Gamma\Rightarrow\varphi \mid\,\Rightarrow$ \ \ or \ $\Gamma\models_L \varphi$ iff BSC-L $\vdash \ \Rightarrow \mid \Gamma\Rightarrow\varphi$

\1

depending on whether $\models_L$ denotes consequence relation for logics characterised by 1-matrices or by 2-matrices. The adequacy of BSC-L for all concrete logics is proved in the same way as for BSC-K$_3$. Therefore we limit our presentation to the systematic characterisation of rules from which the BSC for suitable logics can be composed.

We start with rules for respective unary operations (including \L u\-ka\-sie\-wicz's modalities):

\1

\noindent\begin{tabular}{ll}

	$(\mathord{\mid\neg_H}\mathord{\Rightarrow})$ \ $\dfrac{S \mid \Pi\Rightarrow\Sigma, \varphi}{S \mid \neg\varphi, \Pi\Rightarrow\Sigma}$ &
	$(\mathord{\mid\Rightarrow}\mathord{\neg_H})$ \ $\dfrac{S \mid \varphi, \Pi\Rightarrow\Sigma}{S \mid \Pi\Rightarrow\Sigma, \neg\varphi}$\\[16pt]

	$(\mathord{\neg_B}\mathord{\Rightarrow}\mid)$ \ $\dfrac{\Gamma
		\Rightarrow \Delta, \varphi \mid S}{\neg \varphi, \Gamma \Rightarrow \Delta \mid S}$ &
	$(\mathord{\Rightarrow}\mathord{\neg_B}\mid)$ \ $\dfrac{\varphi, \Gamma
		\Rightarrow \Delta \mid S}{\Gamma
		\Rightarrow \Delta, \neg\varphi\mid S}$\\[16pt]
	
\end{tabular}	

$(\mathord{\mid\neg_P}\mathord{\Rightarrow})$ \ $\dfrac{\Gamma
	\Rightarrow \Delta \mid \Pi\Rightarrow\Sigma, \varphi \qquad \varphi, \Gamma\Rightarrow\Delta \mid \Pi\Rightarrow\Sigma}{\Gamma \Rightarrow \Delta \mid \neg\varphi, \Pi\Rightarrow\Sigma}$ 

\1

\noindent\begin{tabular}{ll}

	$(\mathord{\mid\Rightarrow}\mathord{\neg_P})$ \ $\dfrac{\Gamma
		\Rightarrow \Delta, \varphi \mid \varphi, \Pi\Rightarrow\Sigma}{\Gamma
		\Rightarrow \Delta \mid \Pi\Rightarrow\Sigma, \neg\varphi}$ &
	$(\mathord{\neg_{DP}}\mathord{\Rightarrow\mid})$ \ $\dfrac{\Gamma
		\Rightarrow \Delta, \varphi \mid \varphi, \Pi\Rightarrow\Sigma}{\neg\varphi, \Gamma
		\Rightarrow \Delta \mid \Pi\Rightarrow\Sigma}$
	
\end{tabular}	

\1

$(\mathord{\Rightarrow\neg_{DP}\mid})$ \ $\dfrac{\Gamma
	\Rightarrow \Delta \mid \Pi\Rightarrow\Sigma, \varphi \qquad \varphi, \Gamma\Rightarrow\Delta \mid \Pi\Rightarrow\Sigma}{\Gamma \Rightarrow \Delta, \neg\varphi \mid \Pi\Rightarrow\Sigma}$

\1

The remaining rules in each case (namely $(\mathord{\neg_H}\mathord{\Rightarrow}\!\mid)$, $(\mathord{\Rightarrow}\mathord{\neg_H}\!\mid)$, $(\mathord{\mid\!\neg_B}\mathord{\Rightarrow})$, $(\mathord{\mid\Rightarrow}\mathord{\neg_B})$ $(\mathord{\neg_P}\mathord{\Rightarrow}\mid)$, $(\mathord{\Rightarrow}\mathord{\neg_P}\mid)$, $(\mathord{\mid\neg_{DP}}\mathord{\Rightarrow})$ and $(\mathord{\mid\Rightarrow}\mathord{\neg_{DP}})$) are like the respective rules of
BSC-K$_3$. Consider the premisses of $(\mathord{\mid\Rightarrow}\mathord{\neg_P})$ and $(\mathord{\neg_{DP}}\mathord{\Rightarrow\mid})$ displaying two occurrences of the same side formula. In semantic terms it gives the effect of evaluating $\varphi$ as undefined.

\1

\noindent\begin{tabular}{ll}

	$(\mathord{\Diamond}\mathord{\Rightarrow}\mid)$ \ $\dfrac{\Gamma
		\Rightarrow \Delta \mid \varphi, \Pi\Rightarrow\Sigma}{\Diamond \varphi, \Gamma \Rightarrow \Delta \mid \Pi\Rightarrow\Sigma}$ &
	$(\mathord{\Rightarrow}\mathord{\Diamond}\mid)$ \ $\dfrac{\Gamma
		\Rightarrow \Delta \mid \Pi\Rightarrow\Sigma, \varphi}{\Gamma
		\Rightarrow \Delta, \Diamond\varphi\mid \Pi\Rightarrow\Sigma}$\\[16pt]

	$(\mathord{\mid\Diamond}\mathord{\Rightarrow})$ \ $\dfrac{\Gamma
		\Rightarrow \Delta \mid \varphi, \Pi\Rightarrow\Sigma}{\Gamma
		\Rightarrow \Delta \mid \Diamond\varphi, \Pi\Rightarrow\Sigma}$ &
	$(\mathord{\mid\Rightarrow}\mathord{\Diamond})$ \ $\dfrac{\Gamma
		\Rightarrow \Delta \mid \Pi\Rightarrow\Sigma, \varphi}{\Gamma
		\Rightarrow \Delta \mid \Pi\Rightarrow\Sigma, \Diamond\varphi}$\\[16pt]
	
	$(\mathord{\Box}\mathord{\Rightarrow}\mid)$ \ $\dfrac{\varphi, \Gamma
		\Rightarrow \Delta \mid \Pi\Rightarrow\Sigma}{\Box\varphi, \Gamma \Rightarrow \Delta \mid \Pi\Rightarrow\Sigma}$ &
	$(\mathord{\Rightarrow}\mathord{\Box}\mid)$ \ $\dfrac{\Gamma
		\Rightarrow \Delta, \varphi \mid \Pi\Rightarrow\Sigma}{\Gamma
		\Rightarrow \Delta, \Box\varphi\mid \Pi\Rightarrow\Sigma}$\\[16pt]

	$(\mathord{\mid\Box}\mathord{\Rightarrow})$ \ $\dfrac{\varphi, \Gamma
		\Rightarrow \Delta \mid \Pi\Rightarrow\Sigma}{\Gamma
		\Rightarrow \Delta \mid \Box\varphi, \Pi\Rightarrow\Sigma}$ &
	$(\mathord{\mid\Rightarrow}\mathord{\Box})$ \ $\dfrac{\Gamma
		\Rightarrow \Delta, \varphi \mid \Pi\Rightarrow\Sigma}{\Gamma
		\Rightarrow \Delta \mid \Pi\Rightarrow\Sigma, \Box\varphi}$\\[16pt]

\end{tabular}

Not surprisingly rules introducing modal formula to antecedents or to succedents of 1- and 2-sequents have the same premisses which is a consequence of the fact that such a formula is never undefined. The same remark applies to rules for $\neg_H$ and $\neg_B$.

\1

The set of rules for weak $\wedge, \vee, \rightarrow$ is also partly identical with those for BSC-K$_3$. The identical rules are $(\mathord{\wedge_w}\mathord{\Rightarrow\mid})$, $(\mathord{\Rightarrow}\mathord{\wedge_w\mid})$, $(\mathord{\mid\Rightarrow}\mathord{\vee_w})$, $(\mathord{\mid\vee_w}\mathord{\Rightarrow})$, $(\mathord{\mid\Rightarrow}\mathord{\rightarrow_w})$ and $(\mathord{\mid\rightarrow_w}\mathord{\Rightarrow})$. In the remaining cases we have three premiss rules:

\1

\noindent {\footnotesize $(\mathord{\mid\wedge_w}\mathord{\Rightarrow})$ \ $\dfrac{\Gamma
		\Rightarrow \Delta \mid \varphi, \psi, \Pi\Rightarrow\Sigma \qquad \Gamma\Rightarrow\Delta, \varphi \mid \varphi,
		\Pi \Rightarrow \Sigma \qquad \Gamma\Rightarrow\Delta, \psi \mid \psi, \Pi\Rightarrow\Sigma} {\Gamma
		\Rightarrow \Delta \mid \varphi\wedge\psi, \Pi\Rightarrow\Sigma}$}

\1

\noindent {\footnotesize $(\mathord{\mid\Rightarrow}\mathord{\wedge_w})$\ $\dfrac{\Gamma
		\Rightarrow \Delta \mid \Pi\Rightarrow\Sigma, \varphi, \psi \qquad \varphi, \Gamma\Rightarrow\Delta \mid \Pi \Rightarrow \Sigma, \psi \qquad \psi, \Gamma\Rightarrow\Delta \mid \Pi\Rightarrow\Sigma, \varphi} {\Gamma
		\Rightarrow \Delta \mid \Pi\Rightarrow\Sigma, \varphi\wedge\psi}$}

\1

\noindent {\footnotesize $(\mathord{\Rightarrow}\mathord{\vee_w}\mid)$ \ $\dfrac{\Gamma
		\Rightarrow \Delta, \varphi, \psi \mid \Pi\Rightarrow\Sigma \qquad \Gamma\Rightarrow\Delta, \varphi \mid \varphi,
		\Pi \Rightarrow \Sigma \qquad \Gamma\Rightarrow\Delta, \psi \mid \psi, \Pi\Rightarrow\Sigma} {\Gamma
		\Rightarrow \Delta, \varphi\vee\psi \mid \Pi\Rightarrow\Sigma}$}

\1

\noindent {\footnotesize $(\mathord{\vee_w}\mathord{\Rightarrow}\mid)$ \ $\dfrac{\varphi, \psi, \Gamma
		\Rightarrow \Delta \mid \Pi\Rightarrow\Sigma \qquad \varphi, \Gamma\Rightarrow\Delta \mid \Pi \Rightarrow \Sigma, \psi \qquad \psi, \Gamma\Rightarrow\Delta \mid \Pi\Rightarrow\Sigma, \varphi} {\varphi\vee\psi, \Gamma
		\Rightarrow \Delta \mid \Pi\Rightarrow\Sigma}$}

\1

\noindent {\footnotesize $(\mathord{\Rightarrow}\mathord{\rightarrow_w}\mid)$ \ $\dfrac{\Gamma
		\Rightarrow \Delta, \psi \mid \varphi, \Pi\Rightarrow\Sigma \qquad \Gamma\Rightarrow\Delta, \varphi \mid \varphi,
		\Pi \Rightarrow \Sigma \qquad \Gamma\Rightarrow\Delta, \psi \mid \psi, \Pi\Rightarrow\Sigma} {\Gamma
		\Rightarrow \Delta, \varphi\rightarrow\psi \mid \Pi\Rightarrow\Sigma}$}

\1

\noindent {\footnotesize $(\mathord{\rightarrow_w}\mathord{\Rightarrow}\mid)$ \ $\dfrac{\varphi, \psi, \Gamma
		\Rightarrow \Delta \mid \Pi\Rightarrow\Sigma \qquad \Gamma\Rightarrow\Delta \mid \Pi \Rightarrow \Sigma, \varphi, \psi \qquad \psi, \Gamma\Rightarrow\Delta \mid \Pi\Rightarrow\Sigma, \varphi} {\varphi\rightarrow\psi, \Gamma
		\Rightarrow \Delta \mid \Pi\Rightarrow\Sigma}$}

\1

In the case of {\bf K$_3^\rightarrow$} and {\bf K$_3^\leftarrow$} the specific rules are:

\1

\noindent$(\mid \mathord{\Rightarrow}\mathord{\wedge_{mC}})$ \ $\dfrac{\Gamma
	\Rightarrow \Delta \mid \Pi\Rightarrow\Sigma, \varphi \qquad \varphi, \Gamma\Rightarrow\Delta \mid \Pi\Rightarrow\Sigma, \psi}{\Gamma\Rightarrow\Delta \mid \Pi\Rightarrow\Sigma, \varphi\wedge\psi}$ 

\1

\noindent$(\mid \mathord{\wedge}\mathord{\Rightarrow_{mC}})$ \ $\dfrac{\Gamma\Rightarrow \Delta \mid \varphi, \psi, \Pi\Rightarrow\Sigma \qquad \Gamma\Rightarrow\Delta, \varphi \mid \varphi, \Pi\Rightarrow\Sigma}{\Gamma\Rightarrow\Delta \mid \varphi\wedge\psi, \Pi\Rightarrow\Sigma}$ 

\1

\noindent$(\mathord{\vee}\mathord{\Rightarrow_{mC}}\mid)$ \ $\dfrac{\varphi, \Gamma
	\Rightarrow \Delta \mid \Pi\Rightarrow\Sigma \qquad \psi, \Gamma\Rightarrow\Delta \mid \Pi\Rightarrow\Sigma, \varphi}{\varphi\vee\psi, \Gamma\Rightarrow\Delta \mid \Pi\Rightarrow\Sigma}$ 

\1

\noindent$(\mathord{\Rightarrow}\mathord{\vee_{mC}}\mid)$ \ $\dfrac{\Gamma\Rightarrow \Delta, \varphi, \psi \mid \Pi\Rightarrow\Sigma \qquad \Gamma\Rightarrow\Delta, \varphi \mid \varphi, \Pi\Rightarrow\Sigma}{\Gamma\Rightarrow\Delta, \varphi\vee\psi \mid \Pi\Rightarrow\Sigma}$ 

\1

\noindent$(\mathord{\Rightarrow}\mathord{\rightarrow_{mC}}\mid)$ \ $\dfrac{\Gamma
	\Rightarrow \Delta, \psi \mid \varphi, \Pi\Rightarrow\Sigma \qquad \Gamma\Rightarrow\Delta, \varphi \mid \varphi, \Pi\Rightarrow\Sigma}{\Gamma\Rightarrow\Delta, \varphi\rightarrow\psi \mid \Pi\Rightarrow\Sigma}$ 

\1

\noindent$(\mathord{\rightarrow}\mathord{\Rightarrow_{mC}}\mid)$ \ $\dfrac{\varphi, \psi, \Gamma\Rightarrow \Delta \mid \Pi\Rightarrow\Sigma \qquad \Gamma\Rightarrow\Delta \mid \Pi\Rightarrow\Sigma, \varphi}{\varphi\rightarrow\psi, \Gamma\Rightarrow\Delta \mid \Pi\Rightarrow\Sigma}$ 

\1

\noindent$(\mid \mathord{\Rightarrow}\mathord{\wedge_K})$ \ $\dfrac{\Gamma
	\Rightarrow \Delta \mid \Pi\Rightarrow\Sigma, \psi \qquad \psi, \Gamma\Rightarrow\Delta \mid \Pi\Rightarrow\Sigma, \varphi}{\Gamma\Rightarrow\Delta \mid \Pi\Rightarrow\Sigma, \varphi\wedge\psi}$ 

\1

\noindent$(\mid \mathord{\wedge}\mathord{\Rightarrow_K})$ \ $\dfrac{\Gamma\Rightarrow \Delta \mid \varphi, \psi, \Pi\Rightarrow\Sigma \qquad \Gamma\Rightarrow\Delta, \psi \mid \psi, \Pi\Rightarrow\Sigma}{\Gamma\Rightarrow\Delta \mid \varphi\wedge\psi, \Pi\Rightarrow\Sigma}$ 

\1

\noindent$(\mathord{\vee}\mathord{\Rightarrow_K}\mid)$ \ $\dfrac{\varphi, \Gamma
	\Rightarrow \Delta \mid \Pi\Rightarrow\Sigma, \psi \qquad \psi, \Gamma\Rightarrow\Delta \mid \Pi\Rightarrow\Sigma}{\varphi\vee\psi, \Gamma\Rightarrow\Delta \mid \Pi\Rightarrow\Sigma}$ 

\1

\noindent$(\mathord{\Rightarrow}\mathord{\vee_K}\mid)$ \ $\dfrac{\Gamma\Rightarrow \Delta, \varphi, \psi \mid \Pi\Rightarrow\Sigma \qquad \Gamma\Rightarrow\Delta, \psi \mid \psi, \Pi\Rightarrow\Sigma}{\Gamma\Rightarrow\Delta, \varphi\vee\psi \mid \Pi\Rightarrow\Sigma}$ 

\1

\noindent$(\mathord{\Rightarrow}\mathord{\rightarrow_K}\mid)$ \ $\dfrac{\Gamma
	\Rightarrow \Delta, \psi \mid \varphi, \Pi\Rightarrow\Sigma \qquad \Gamma\Rightarrow\Delta, \psi \mid \psi, \Pi\Rightarrow\Sigma}{\Gamma\Rightarrow\Delta, \varphi\rightarrow\psi \mid \Pi\Rightarrow\Sigma}$ 

\1

\noindent$(\mathord{\rightarrow}\mathord{\Rightarrow_K}\mid)$ \ $\dfrac{\psi, \Gamma\Rightarrow \Delta \mid \Pi\Rightarrow\Sigma \qquad \Gamma\Rightarrow\Delta \mid \Pi\Rightarrow\Sigma, \varphi, \psi}{\varphi\rightarrow\psi, \Gamma\Rightarrow\Delta \mid \Pi\Rightarrow\Sigma}$ 

\1

The remaining rules in both cases are identical with  $(\mathord{\wedge}\mathord{\Rightarrow\mid})$, $(\mathord{\Rightarrow}\mathord{\wedge\mid})$, $(\mathord{\mid\Rightarrow}\mathord{\vee})$, $(\mathord{\mid\vee}\mathord{\Rightarrow})$, $(\mid\mathord{\Rightarrow}\mathord{\rightarrow})$ and $(\mid\mathord{\rightarrow}\mathord{\Rightarrow})$ from BSC-K$_3$.

\1

The implications of \L ukasiewicz \cite{Lukasiewicz} and his specific additive $\wedge_L$ and $\vee_L$ are characterised by the following rules:

\begin{flushleft}
	$(\mid\mathord{\wedge_L\Rightarrow})$ \ $\dfrac{\varphi, \Gamma
		\Rightarrow \Delta \mid \psi, \Pi\Rightarrow\Sigma \qquad \psi, \Gamma\Rightarrow\Delta \mid \varphi, \Pi\Rightarrow\Sigma}{ \Gamma\Rightarrow\Delta \mid \varphi\wedge\psi,\Pi\Rightarrow\Sigma}$
\end{flushleft} 

\begin{flushleft}
	{\footnotesize $(\mid\mathord{\Rightarrow\wedge_L})$ \ $ \dfrac{\Gamma
			\Rightarrow \Delta, \varphi, \psi \mid \Pi\Rightarrow\Sigma \qquad\varphi, \Gamma\Rightarrow\Delta \mid
			\Pi \Rightarrow \Sigma, \psi \qquad \psi, \Gamma\Rightarrow\Delta \mid \Pi\Rightarrow\Sigma, \varphi } {\Gamma
			\Rightarrow \Delta\mid \Pi\Rightarrow\Sigma, \varphi\wedge\psi }$}
\end{flushleft}

\begin{flushleft}
	$(\mathord{\mid\vee_L\Rightarrow})$ \ $\dfrac{\varphi, \Gamma
		\Rightarrow \Delta \mid \psi, \Pi\Rightarrow\Sigma \qquad \psi, \Gamma\Rightarrow\Delta\mid \varphi, \Pi\Rightarrow\Sigma}{\Gamma\Rightarrow\Delta \mid \varphi\vee\psi, \Pi\Rightarrow\Sigma}$
\end{flushleft} 

\begin{flushleft}
	{\footnotesize $(\mathord{\mid\Rightarrow\vee_L})$ \ $\dfrac{\Gamma
			\Rightarrow \Delta, \varphi, \psi \mid \Pi\Rightarrow\Sigma \qquad\varphi,  \Gamma\Rightarrow\Delta \mid
			\Pi \Rightarrow \Sigma, \psi \qquad \psi, \Gamma\Rightarrow\Delta \mid \Pi\Rightarrow\Sigma, \varphi } {\Gamma
			\Rightarrow \Delta \mid \Pi\Rightarrow\Sigma, \varphi\vee\psi}$}
\end{flushleft}

\begin{flushleft}
	$(\mathord{\Rightarrow}\mathord{\rightarrow_L}\mid)$ \ $\dfrac{\varphi, \Gamma
		\Rightarrow \Delta, \psi \mid \Pi\Rightarrow\Sigma \qquad \Gamma\Rightarrow\Delta \mid \varphi, \Pi\Rightarrow\Sigma, \psi}{\Gamma\Rightarrow\Delta, \varphi\rightarrow\psi \mid \Pi\Rightarrow\Sigma}$ 
\end{flushleft}

\begin{flushleft}
	{\small $(\mathord{\rightarrow_L}\mathord{\Rightarrow}\mid)$ \ $\dfrac{\Gamma
			\Rightarrow \Delta, \varphi \mid \psi, \Pi\Rightarrow\Sigma \qquad \psi, \Gamma\Rightarrow\Delta \mid
			\Pi \Rightarrow \Sigma \qquad \Gamma\Rightarrow\Delta \mid \Pi\Rightarrow\Sigma, \varphi } {\varphi\rightarrow\psi, \Gamma
			\Rightarrow \Delta \mid \Pi\Rightarrow\Sigma}$}
\end{flushleft}

The remaining  rules are identical with  $(\mathord{\wedge}\mathord{\Rightarrow\mid})$, $(\mathord{\Rightarrow}\mathord{\wedge\mid})$, $(\mathord{\Rightarrow}\mathord{\vee\mid})$, $(\mathord{\vee}\mathord{\Rightarrow\mid})$, $(\mid\mathord{\Rightarrow}\mathord{\rightarrow})$ and $(\mid\mathord{\rightarrow}\mathord{\Rightarrow})$ from BSC-K$_3$.

\1

For Soboci\'nski's connectives we have:

\begin{flushleft}
	$(\mathord{\wedge_S\Rightarrow}\mid)$ \ $\dfrac{\varphi, \Gamma
		\Rightarrow \Delta \mid \psi, \Pi\Rightarrow\Sigma \qquad \psi, \Gamma\Rightarrow\Delta \mid \varphi, \Pi\Rightarrow\Sigma}{\varphi\wedge\psi, \Gamma\Rightarrow\Delta \mid \Pi\Rightarrow\Sigma}$
\end{flushleft} 

\begin{flushleft}
	{\footnotesize $(\mathord{\Rightarrow\wedge_S}\mid)$ \ $\dfrac{\Gamma
			\Rightarrow \Delta, \varphi, \psi \mid \Pi\Rightarrow\Sigma \qquad\varphi, \Gamma\Rightarrow\Delta \mid
			\Pi \Rightarrow \Sigma, \psi \qquad \psi, \Gamma\Rightarrow\Delta \mid \Pi\Rightarrow\Sigma, \varphi } {\Gamma
			\Rightarrow \Delta, \varphi\wedge\psi \mid \Pi\Rightarrow\Sigma}$}
	
\end{flushleft}
\begin{flushleft}
	$(\mathord{\mid\Rightarrow}\mathord{\vee_S})$ \ $\dfrac{\Gamma
		\Rightarrow \Delta, \varphi \mid \Pi\Rightarrow\Sigma, \psi \qquad \Gamma\Rightarrow\Delta, \psi\mid \Pi\Rightarrow\Sigma, \varphi}{\Gamma\Rightarrow\Delta \mid \Pi\Rightarrow\Sigma, \varphi\vee\psi}$
\end{flushleft} 

\begin{flushleft}
	{\footnotesize $(\mathord{\mid \vee_S}\mathord{\Rightarrow})$ \ $\dfrac{\Gamma
			\Rightarrow \Delta \mid \varphi, \psi, \Pi\Rightarrow\Sigma \qquad\varphi,  \Gamma\Rightarrow\Delta \mid
			\Pi \Rightarrow \Sigma, \psi \qquad \psi, \Gamma\Rightarrow\Delta \mid \Pi\Rightarrow\Sigma, \varphi } {\Gamma
			\Rightarrow \Delta \mid \varphi\vee\psi, \Pi\Rightarrow\Sigma}$}
\end{flushleft}

\begin{flushleft}
	$(\mathord{\mid\Rightarrow}\mathord{\rightarrow_S})$ \ $\dfrac{\varphi, \Gamma
		\Rightarrow \Delta, \psi \mid \Pi\Rightarrow\Sigma \qquad \Gamma\Rightarrow\Delta \mid \varphi, \Pi\Rightarrow\Sigma, \psi}{\Gamma\Rightarrow\Delta \mid \Pi\Rightarrow\Sigma, \varphi\rightarrow\psi}$
\end{flushleft} 

\begin{flushleft}
	{\small $(\mathord{\mid \rightarrow_S}\mathord{\Rightarrow})$ \ $\dfrac{\Gamma
			\Rightarrow \Delta, \varphi \mid \psi, \Pi\Rightarrow\Sigma \qquad\psi, \Gamma\Rightarrow\Delta \mid
			\Pi \Rightarrow \Sigma \qquad \Gamma\Rightarrow\Delta \mid \Pi\Rightarrow\Sigma, \varphi} {\Gamma
			\Rightarrow \Delta \mid \varphi\rightarrow\psi, \Pi\Rightarrow\Sigma}$}
\end{flushleft}

The remaining rules are again like those in BSC-K$_3$. 

\1

Sette's connectives are characterised by the following rules:

\begin{flushleft}
		$(\mathord{\wedge_{\mathit{Se}}}\mathord{\Rightarrow\mid})$  $\dfrac{
			\Gamma \Rightarrow \Delta \mid \varphi, \psi, \Pi\Rightarrow\Sigma}{\varphi\wedge\psi, \Gamma
			\Rightarrow \Delta \mid \Pi\Rightarrow\Sigma}$
\end{flushleft}	\begin{flushleft}
$(\mathord{\Rightarrow}\mathord{\wedge_{\mathit{Se}}\mid})$  $\dfrac{\Gamma
			\Rightarrow \Delta \mid \Pi\Rightarrow\Sigma, \varphi \;\quad
			\Gamma\Rightarrow \Delta \mid \Pi\Rightarrow\Sigma, \psi}{\Gamma \Rightarrow \Delta,
			\varphi\wedge\psi \mid \Pi\Rightarrow\Sigma}$
\end{flushleft}
	
\begin{flushleft}
	 $(\mathord{\Rightarrow}\mathord{\vee_{\mathit{Se}}\mid})$ $\dfrac{\Gamma
			\Rightarrow \Delta \mid \Pi\Rightarrow\Sigma, \varphi,\psi } {\Gamma\Rightarrow \Delta,
			\varphi\vee\psi \mid \Pi\Rightarrow\Sigma}$
\end{flushleft}
		
\begin{flushleft}
			 $(\mathord{\vee_{\mathit{Se}}}\mathord{\Rightarrow\mid})$  $\dfrac{\Gamma
			\Rightarrow \Delta \mid \varphi, \Pi\Rightarrow\Sigma \;\quad \Gamma
			\Rightarrow \Delta \mid \psi, \Pi\Rightarrow\Sigma}{\varphi\vee\psi,
			\Gamma\Rightarrow \Delta \mid\Pi\Rightarrow\Sigma}$
\end{flushleft}
	
\begin{flushleft}
	 $(\mathord{\Rightarrow}\mathord{\rightarrow_{\mathit{Se}}\mid})$ $\dfrac{\Gamma
			\Rightarrow \Delta \mid \varphi, \Pi\Rightarrow\Sigma, \psi}{\Gamma\Rightarrow\Delta, \varphi\rightarrow\psi \mid \Pi\Rightarrow\Sigma}$
\end{flushleft} \begin{flushleft}
$(\mathord{\rightarrow_{\mathit{Se}}}\mathord{\Rightarrow\mid})$ $\dfrac{\Gamma
			\Rightarrow \Delta \mid \Pi\Rightarrow\Sigma, \varphi \;\quad \Gamma\Rightarrow\Delta \mid \psi,
			\Pi \Rightarrow \Sigma } {\varphi\rightarrow\psi, \Gamma
			\Rightarrow \Delta \mid \Pi\Rightarrow\Sigma}$
\end{flushleft}
	
\begin{flushleft}
	 $(\mathord{\mid\Rightarrow}\mathord{\rightarrow_{Se}})$ $\dfrac{S \mid \varphi, \Pi\Rightarrow\Sigma, \psi}{S \mid \Pi\Rightarrow\Sigma, \varphi\rightarrow\psi}$
\end{flushleft} \begin{flushleft}
$(\mathord{\mid\rightarrow_{Se}}\mathord{\Rightarrow})$ $\dfrac{S \mid \Pi\Rightarrow\Sigma, \varphi \qquad S \mid \psi,
		\Pi \Rightarrow \Sigma } {S \mid \varphi\rightarrow\psi, \Pi\Rightarrow\Sigma}$
\end{flushleft}

$(\mathord{\mid\wedge_{Se}}\mathord{\Rightarrow})$,
$(\mathord{\mid\Rightarrow}\mathord{\wedge_{Se}})$, $(\mathord{\mid\Rightarrow}\mathord{\vee_{Se}})$,
$(\mathord{\mid\vee_{Se}}\mathord{\Rightarrow})$ like in BSC-K$_3$.

\1

Finally, Carnielli's and Sette's connectives characterising ${\bf I^1}$ and ${\bf I^2}$:

\begin{flushleft}
	 $(\mathord{{\mid}\wedge_C}\mathord{\Rightarrow})$  $\dfrac{\varphi, \psi,
			\Gamma \Rightarrow \Delta \mid \Pi\Rightarrow\Sigma}{ \Gamma
			\Rightarrow \Delta \mid \varphi\wedge\psi, \Pi\Rightarrow\Sigma}$
\end{flushleft} \begin{flushleft}
$(\mathord{\mid\Rightarrow}\mathord{\wedge_C})$  $\dfrac{\Gamma
			\Rightarrow \Delta, \varphi \mid \Pi\Rightarrow\Sigma \quad
			\Gamma\Rightarrow \Delta, \psi \mid \Pi\Rightarrow\Sigma}{\Gamma \Rightarrow \Delta \mid \Pi\Rightarrow\Sigma, \varphi\wedge\psi}$
\end{flushleft}
	
\begin{flushleft}
	 $(\mathord{\mid\Rightarrow}\mathord{\vee_C})$	  $\dfrac{\Gamma
			\Rightarrow \Delta, \varphi, \psi \mid \Pi\Rightarrow\Sigma }{		\Gamma\Rightarrow \Delta \mid\Pi\Rightarrow\Sigma, \varphi\vee\psi}$
\end{flushleft} \begin{flushleft}
$(\mathord{{\mid}\vee_C}\mathord{\Rightarrow})$	  $\dfrac{\varphi, \Gamma
			\Rightarrow \Delta \mid \Pi\Rightarrow\Sigma \quad \psi, \Gamma
			\Rightarrow \Delta \mid \Pi\Rightarrow\Sigma}
		{\Gamma\Rightarrow \Delta \mid
			\varphi\vee\psi, \Pi\Rightarrow\Sigma}$
\end{flushleft}
\begin{flushleft}
	 $(\mathord{\mid\Rightarrow}\mathord{\rightarrow_C})$ $\dfrac{\varphi, \Gamma
			\Rightarrow \Delta, \psi \mid \Pi\Rightarrow\Sigma}{\Gamma\Rightarrow\Delta \mid \Pi\Rightarrow\Sigma, \varphi\rightarrow\psi}$
\end{flushleft} \begin{flushleft}
$(\mathord{\mid\rightarrow_C}\mathord{\Rightarrow})$ $\dfrac{\Gamma
			\Rightarrow \Delta, \varphi \mid \Pi\Rightarrow\Sigma \quad \psi, \Gamma\Rightarrow\Delta \mid \Pi \Rightarrow \Sigma } {\Gamma
			\Rightarrow \Delta \mid \varphi\rightarrow\psi, \Pi\Rightarrow\Sigma}$
\end{flushleft}

\begin{flushleft}
		$(\mathord{\Rightarrow}\mathord{\rightarrow_C\mid})$ \ $\dfrac{\varphi, \Gamma
		\Rightarrow \Delta, \psi \mid S}{\Gamma\Rightarrow\Delta, \varphi\rightarrow\psi \mid S}$
\end{flushleft} 
\begin{flushleft}
		$(\mathord{\rightarrow_C}\mathord{\Rightarrow\mid})$ \ $\dfrac{\Gamma
		\Rightarrow \Delta, \varphi \mid S \qquad\psi,
		\Gamma \Rightarrow \Delta \mid S } {\varphi\rightarrow\psi, \Gamma
		\Rightarrow \Delta \mid S}$
\end{flushleft}

$(\mathord{\wedge_C}\mathord{\Rightarrow}\mid)$,
$(\mathord{\Rightarrow}\mathord{\wedge_C}\mid)$, $(\mathord{\Rightarrow}\mathord{\vee_C}\mid)$,
$(\mathord{\vee_C}\mathord{\Rightarrow}\mid)$ like in BSC-K$_3$. 

\1

We finish with the characterisation of the remaining implications introduced in section 2. In most cases they are obtained by combining rules which were previously introduced. In particular:

S\l upecki's \cite{Slupecki67} implication is characterised by means of:
$(\mathord{\mid\Rightarrow}\mathord{\rightarrow})$ and $(\mathord{\mid \rightarrow}\mathord{\Rightarrow})$ from BSC-K$_3$ as well as $(\mathord{\Rightarrow}\mathord{\rightarrow_C}\mid)$ and $(\mathord{\rightarrow_C}\mathord{\Rightarrow}\mid)$.

Heyting's implication \cite{Heyting} is characterised by means of:
$(\mathord{\rightarrow_L}\mathord{\Rightarrow}\mid)$, $(\mathord{\Rightarrow}\mathord{\rightarrow_L}\mid)$, $(\mathord{\mid\Rightarrow}\mathord{\rightarrow_{Se}})$,$(\mathord{\mid \rightarrow_{Se}}\mathord{\Rightarrow})$.

D'Ottaviano's/da Costa's/Ja\'{s}kowski's/S\l{}upecki's implication \cite{DOttavianodaCosta,Jaskowski48,Slupecki39} is characterised by means of:
$(\mathord{\rightarrow}\mathord{\Rightarrow}\mid)$, $(\mathord{\Rightarrow}\mathord{\rightarrow}\mid)$, $(\mathord{\mid\Rightarrow}\mathord{\rightarrow_{Se}})$,$(\mathord{\mid \rightarrow_{Se}}\mathord{\Rightarrow})$.

Rescher's implication \cite{Rescher} is characterised by means of:
$(\mathord{\rightarrow_L}\mathord{\Rightarrow}\mid)$, $(\mathord{\Rightarrow}\mathord{\rightarrow_L}\mid)$, $(\mathord{\mid\Rightarrow}\mathord{\rightarrow_{S}})$,$(\mathord{\mid \rightarrow_{S}}\mathord{\Rightarrow})$.

Tomova's implication \cite{Heyting} is characterised by means of:
$(\mathord{\rightarrow_L}\mathord{\Rightarrow}\mid)$, $(\mathord{\Rightarrow}\mathord{\rightarrow_L}\mid)$, $(\mathord{\mid\Rightarrow}\mathord{\rightarrow_{C}})$,$(\mathord{\mid \rightarrow_{C}}\mathord{\Rightarrow})$.

Only in the case of Soboci\'nski's implication $\rightarrow_S'$ do we have a pair of new rules:

\begin{flushleft}
	{\footnotesize $(\mathord{\rightarrow_S'}\mathord{\Rightarrow\mid})$ \ $\dfrac{\psi, \Gamma
			\Rightarrow \Delta \mid \Pi\Rightarrow\Sigma \qquad \Gamma\Rightarrow\Delta \mid \Pi \Rightarrow \Sigma, \varphi, \psi \qquad \Gamma\Rightarrow\Delta, \varphi\mid \varphi, \psi, \Pi\Rightarrow\Sigma} {\varphi\rightarrow\psi, \Gamma
			\Rightarrow \Delta \mid \Pi\Rightarrow\Sigma}$}
\end{flushleft}

\begin{flushleft}
	{\footnotesize $(\mathord{\Rightarrow}\mathord{\rightarrow_S'\mid})$ \ $\dfrac{\varphi, \Gamma
			\Rightarrow \Delta, \psi \mid \Pi\Rightarrow\Sigma \;\quad \Gamma\Rightarrow\Delta \mid\varphi,
			\Pi \Rightarrow \Sigma, \psi \;\quad \Gamma\Rightarrow\Delta, \psi \mid \psi, \Pi\Rightarrow\Sigma, \varphi } {\Gamma
			\Rightarrow \Delta, \varphi\rightarrow\psi \mid \Pi\Rightarrow\Sigma}$}
\end{flushleft}

The remaining two rules are: $(\mathord{\mid\Rightarrow}\mathord{\rightarrow_{Se}})$ and $(\mathord{\mid \rightarrow_{Se}}\mathord{\Rightarrow})$.

\1

Now, consider an arbitrary connective $c$ of the logic {\bf L}, the corresponding operation $\underline{c}$ as characterised by suitable matrix determining {\bf L} in section 2, and the four rules for $c$. The following holds:

\begin{theorem} 	
	For all presented rules characterising arbitrary $c$ of any {\bf L}: all premisses are {\bf L}-valid iff the conclusion is {\bf L}-valid.
\end{theorem}

\begin{proof}
This is an analogue of Theorem \ref{Theorem1} for any considered logic {\bf L} which implies the adequacy of BSC-L. We provide only one example for illustration. Consider the rule $(\mathord{\Rightarrow}\mathord{\rightarrow_S'\mid})$. Suppose that $ {\bf L}\models \varphi,\Gamma\Rightarrow\Delta,\psi\mid\Pi\Rightarrow\Sigma$, $ {\bf L}\models \Gamma\Rightarrow\Delta\mid\varphi,\Pi\Rightarrow\Sigma,\psi$, $ {\bf L}\models \Gamma\Rightarrow\Delta,\psi\mid\psi,\Pi\Rightarrow\Sigma,\varphi$, while $ {\bf L}\not\models \Gamma\Rightarrow\Delta,\varphi\rightarrow\psi\mid\Pi\Rightarrow\Sigma$. Then: 

(1) for each $ h $ there is a $ \chi $ such that ($ \chi\in\Gamma $ and $ h(\chi)\not=1 $) or ($ \chi\in\Delta $ and $ h(\chi)=1 $) or ($ \chi\in\Pi $ and $ h(\chi)=0$) or ($ \chi\in\Sigma $ and $ h(\chi)\not=0 $) or ($ \chi=\varphi $ and $ h(\varphi)\not=1 $) or ($ \chi=\psi $ and $ h(\psi)=1 $); 

(2) for each $ h $ there is an $ \omega $ such that ($ \omega\in\Gamma $ and $ h(\omega)\not=1 $) or ($ \omega\in\Delta $ and $ h(\omega)=1 $) or ($ \omega\in\Pi $ and $ h(\omega)=0$) or ($ \omega\in\Sigma $ and $ h(\omega)\not=0 $) or ($ \omega=\varphi $ and $ h(\varphi)=0 $) or ($ \omega=\psi $ and $ h(\psi)\not=0 $); 

(3) for each $ h $ there is a $ \tau $ such that ($ \tau\in\Gamma $ and $ h(\tau)\not=1 $) or ($ \tau\in\Delta $ and $ h(\tau)=1 $) or ($ \tau\in\Pi $ and $ h(\tau)=0$) or ($ \tau\in\Sigma $ and $ h(\tau)\not=0 $) or ($ \tau=\varphi $ and $ h(\varphi)\not=0 $) or ($ \tau=\psi $ and ($ h(\psi)=1 $ or $ h(\psi)=0 $)); 

(4) for some $ h $, it holds that, for each $ \gamma\in\Gamma $, $h(\gamma)=1$, for each $ \delta\in\Delta $, $ h(\delta)\not=1 $, for each $ \pi\in\Pi $, $ h(\pi)\not=0 $, for each $ \sigma\in\Sigma $, $ h(\sigma)=0 $, and $ h(\varphi\rightarrow\psi)\not=1 $.

Therefore, (5) $ h(\varphi)\not=1 $ or $ h(\psi)=1 $, (6) $ h(\varphi)=0 $ or $ h(\psi)\not=0 $, (7) $ h(\varphi)\not=0 $ or ($ h(\psi)=1 $ or $ h(\psi)=0 $), i.e. $ h(\varphi)\not=0 $ or $ h(\psi)\not=\mathrm{u} $, and (8) $ h(\varphi\rightarrow\psi)\not=1 $. Suppose that (9) $ h(\varphi\rightarrow\psi)=\mathrm{u} $. Then either (10) $ h(\varphi)=1 $ and $ h(\psi)=\mathrm{u} $ or (11) $ h(\varphi)=0 $ and $ h(\psi)=\mathrm{u} $. Hence, (12) $ h(\psi)=\mathrm{u} $ and ($ h(\varphi)=1 $ or $ h(\varphi)=0 $). From (5), (7), and (12) we obtain (13) $ h(\varphi)\not=1 $ and $ h(\varphi)\not=0 $. Contradiction. Suppose that (14) $ h(\varphi\rightarrow\psi)=0$. Then either (15) $ h(\varphi)=1 $ and $ h(\psi)=0 $ or (16) $ h(\varphi)=\mathrm{u} $ and $ h(\psi)=0 $. Thus, (17) $ h(\psi)=0 $ and ($ h(\varphi)=1 $ or $ h(\varphi)=\mathrm{u} $), i.e. $ h(\psi)=0 $ and $ h(\varphi)\not=0 $. From (6) and (17) we get $ h(\varphi)=0 $. Contradiction. Consequently, $ {\bf L}\models \Gamma\Rightarrow\Delta,\varphi\rightarrow\psi\mid\Pi\Rightarrow\Sigma$.

Suppose that $ {\bf L}\models \Gamma\Rightarrow\Delta,\varphi\rightarrow\psi\mid\Pi\Rightarrow\Sigma$, while $ {\bf L}\not\models \varphi,\Gamma\Rightarrow\Delta,\psi\mid\Pi\Rightarrow\Sigma$ or $ {\bf L}\not\models \Gamma\Rightarrow\Delta\mid\varphi,\Pi\Rightarrow\Sigma,\psi$ or $ {\bf L}\not\models \Gamma\Rightarrow\Delta,\psi\mid\psi,\Pi\Rightarrow\Sigma,\varphi$. Let ($ \ast $) stands for the sentence `for each $ \gamma\in\Gamma $, $h(\gamma)=1$, for each $ \delta\in\Delta $, $ h(\delta)\not=1 $, for each $ \pi\in\Pi $, $ h(\pi)\not=0 $, for each $ \sigma\in\Sigma $, $ h(\sigma)=0 $'. Then:

(1) for each $ h $ there is a $ \chi $ such that ($ \chi\in\Gamma $ and $ h(\chi)\not=1 $) or ($ \chi\in\Delta $ and $ h(\chi)=1 $) or ($ \chi\in\Pi $ and $ h(\chi)=0$) or ($ \chi\in\Sigma $ and $ h(\chi)\not=0 $) or ($ \chi=\varphi\rightarrow\psi $ and $ h(\varphi\rightarrow\psi)=1 $);

(2a) for some $ h $, ($ \ast $), $ h(\varphi)=1$, and $ h(\psi)\not=1$; or 

(2b) for some $ h $, ($ \ast $), $ h(\varphi)\not=0$, and $ h(\psi)=0$; or 

(2c) for some $ h $, ($ \ast $), $ h(\varphi)=0$, and $ h(\psi)=\mathrm{u}$.

It easy to check that

(3a) if $ h(\varphi)=1$ and $ h(\psi)\not=1$, then $ h(\varphi\rightarrow\psi)\not=1 $;

(3b) if $ h(\varphi)\not=0$ and $ h(\psi)=0$, then $ h(\varphi\rightarrow\psi)=0$;

(3c) if $ h(\varphi)=0$ and $ h(\psi)=\mathrm{u}$, then $ h(\varphi\rightarrow\psi)=\mathrm{u}$.

Thus, we obtain a contradiction. Hence, $ {\bf L}\models \varphi,\Gamma\Rightarrow\Delta,\psi\mid\Pi\Rightarrow\Sigma$, $ {\bf L}\models \Gamma\Rightarrow\Delta\mid\varphi,\Pi\Rightarrow\Sigma,\psi$, and $ {\bf L}\models \Gamma\Rightarrow\Delta,\psi\mid\psi,\Pi\Rightarrow\Sigma,\varphi$. 
\end{proof}

\section{Comments on the Applicability of Bisequents}\label{Comments}

In this section we try to substantiate our claim concerning the uniformity of BSC with respect to three-valued logics. First of all we illustrate the method lying behind the construction of rules in BSC. Then we show that it allows us to obtain a finite characterisation of arbitrary three-valued operations (even those which are not axiomatisable in the standard sense), and that we can also express in BSC different notions of the consequence relation, considered in the works devoted to paraconsistent logics.

The uniformity of BSC follows from the fact that the rules are
not computed on the basis of any normal (disjunctive or conjunctive) form,
like in other approaches, but on the basis of geometrical insights connected with the tabular representation of the
respective connective. Geometrical insights are appropriate here: to establish the
premisses for the rule with the principal formula in one of the four positions in
a bisequent, we just examine its tabular representation. For
example, if the indicated values of the arguments form a rectangle, one premiss is
enough whereas in the case of more complex shapes two or three or even four premisses are
required.
Since the process of construction of rules on the basis of tables is not deterministic
we do not propose any algorithm to that end. More detailed general remarks concerning the process of constructing rules can be found in \cite{Indrzejczak3}. In the following we illustrate 
the process of construction with the example of Soboci\'nski's implication $ \rightarrow_S $ which shows several typical strategies of how to encode the content of matrices in BSC rules.

Consider first the simplest rule:

\begin{flushleft}
	$(\mathord{\Rightarrow}\mathord{\rightarrow_S}\mid)$ \, $\dfrac{\Gamma
		\Rightarrow \Delta, \psi \mid \varphi, \Pi\Rightarrow\Sigma}{\Gamma\Rightarrow\Delta, \varphi\rightarrow\psi \mid \Pi\Rightarrow\Sigma}$
\end{flushleft} 

The occurrence of $\varphi\rightarrow\psi$ happens in the succedent of the 1-sequent of the conclusion of the rule $(\mathord{\Rightarrow}\mathord{\rightarrow_S}\mid)$.
The corresponding partial matrix is as follows:

\1

\begin{tabular}{|c|ccc|}
	\hline
	$\underline{\rightarrow_S} $ & 1 & u & 0 \\
	\hline
	1 &  & 0 & 0 \\
	u &  & u & 0 \\
	0 &  &  &  \\
	\hline
\end{tabular} 

\1

We have $\varphi\rightarrow\psi$ is 0 or $u$ iff ($\varphi$ is 1 or $u$) \textit{and} ($\psi$ is $u$ or 0).  The formula $ \psi $ possesses the same semantic values as $\varphi\rightarrow\psi$ hence is also present in the succedent of the 1-sequent, but in the premiss. The formula $ \varphi $ (true or undefined) occurs in the antecedent of the 2-sequent of the premiss.

Let us consider the following rule:

\begin{flushleft}
	$(\mathord{\rightarrow_S}\mathord{\Rightarrow}\mid)$ \, $\dfrac{\Gamma
		\Rightarrow \Delta \mid \Pi\Rightarrow\Sigma, \varphi \qquad\psi, \Gamma\Rightarrow\Delta \mid
		\Pi \Rightarrow \Sigma } {\varphi\rightarrow\psi, \Gamma
		\Rightarrow \Delta \mid \Pi\Rightarrow\Sigma}$
\end{flushleft} 

The occurrence of $\varphi\rightarrow\psi$ happens in the antecedent of the 1-sequent of the conclusion of the rule $(\mathord{\rightarrow_S}\mathord{\Rightarrow}\mid)$.
The corresponding partial matrix is as follows:

\1

\begin{tabular}{|c|ccc|}
	\hline
	$\underline{\rightarrow_S} $ & 1 & u & 0 \\
	\hline
	1 & 1 &  &  \\
	u & 1 &  &  \\
	0 & 1 & 1 & 1 \\
	\hline
\end{tabular}

\1

We have $\varphi\rightarrow\psi$ is 1 iff $\varphi$ is 0 \textit{or} $\psi$ is 1.  The formula $ \psi $ possesses the same semantic values as $\varphi\rightarrow\psi$. It is also present in the antecedent of the 1-sequent, in the right premiss. The formula $ \varphi $ is 0, so it occurs in the succedent of the 2-sequent of the left premiss. Notice that we have two premisses due to the `\textit{or}' in the semantic condition.

Let us consider the following rule:

\begin{flushleft}
	$(\mathord{\mid\Rightarrow}\mathord{\rightarrow_S})$ \ $\dfrac{\varphi, \Gamma
		\Rightarrow \Delta, \psi \mid \Pi\Rightarrow\Sigma \qquad \Gamma\Rightarrow\Delta \mid \varphi, \Pi\Rightarrow\Sigma, \psi}{\Gamma\Rightarrow\Delta \mid \Pi\Rightarrow\Sigma, \varphi\rightarrow\psi}$
\end{flushleft}  

The occurrence of $\varphi\rightarrow\psi$ happens in the succedent of the 2-sequent of the conclusion of the rule $(\mathord{\mid\Rightarrow}\mathord{\rightarrow_S})$.
The corresponding partial matrix is as follows:

\1

\begin{tabular}{|c|ccc|}
	\hline
	$\underline{\rightarrow_S} $ & 1 & u & 0 \\
	\hline
	1 &  & 0 & 0 \\
	u &  &  & 0 \\
	0 &  &  &  \\
	\hline
\end{tabular} 

\1

We have $\varphi\rightarrow\psi$ is 0 iff [$\varphi$ is 1 \textit{and} ($\psi$ is $u$ or 0)] \textit{or} [$\psi$ is 0 \textit{and} ($\varphi$ is 1 or $u$)].  We have two premisses due to the `\textit{or}' in the semantic condition. The left one is determined by the condition `$\varphi$ is 1 \textit{and} ($\psi$ is $u$ or 0)'. So $ \varphi $ occurs in the antecedent of the 1-sequent of the left premiss and $ \psi $ occurs in its succedent. The right premiss is determined by the condition `$\psi$ is 0 \textit{and} ($\varphi$ is 1 or $u$)'. The occurrences of $ \varphi $ and $ \psi$ are similar to those in the left premiss, except that they appear in the 2-sequent, not in the 1-sequent.

Let us consider the last rule:

\begin{flushleft}
	{\small $(\mathord{\mid \rightarrow_S}\mathord{\Rightarrow})$ \ $\dfrac{\Gamma
			\Rightarrow \Delta, \varphi \mid \psi, \Pi\Rightarrow\Sigma \quad\psi, \Gamma\Rightarrow\Delta \mid
			\Pi \Rightarrow \Sigma \quad \Gamma\Rightarrow\Delta \mid \Pi\Rightarrow\Sigma, \varphi} {\Gamma
			\Rightarrow \Delta \mid \varphi\rightarrow\psi, \Pi\Rightarrow\Sigma}$}
\end{flushleft} 

The corresponding partial matrix is as follows:

\1

\begin{tabular}{|c|ccc|}
	\hline
	$\underline{\rightarrow_S} $ & 1 & u & 0 \\
	\hline
	1 & 1 &  &  \\
	u & 1 & u &  \\
	0 & 1 & 1 & 1 \\
	\hline
\end{tabular} 

\1

We have $\varphi\rightarrow\psi$ is 1 or $u$ iff: $\varphi$ is 0 \textit{or} $\psi$ is \textit{1} \textit{or} [($\varphi$ is $u$ or 0) \textit{and} ($\psi$ is 1 or $u$)]. The condition `$\varphi$ is 0' is expressed by the rightmost premiss, where $\varphi$ occurs in the succedent of the 2-sequent.  The condition `$\psi$ is 1' is expressed by the middle premiss, where $\psi$ occurs in the antecedent of the 1-sequent. The condition `($\varphi$ is $u$ or 0) \textit{and} ($\psi$ is 1 or $u$)' is expressed by the leftmost premiss, where $\varphi$ occurs in the succedent of the 1-sequent and $\psi$ occurs in the antecedent of the 2-sequent. In particular, the leftmost premiss corresponds to the following matrix:

\1

\begin{tabular}{|c|ccc|}
	\hline
	$\underline{\rightarrow_S} $ & 1 & u & 0 \\
	\hline
	1 &  &  &  \\
	u & 1 & u &  \\
	0 & 1 & 1 &  \\
	\hline
\end{tabular}

\1

In fact we can consider in this way arbitrary three-valued operations. Let us take into account two remarkable connectives due to Pa\l{}asi\'{n}ska \cite{Palasinska} (originally, she gave them numbers 7 and 8). 

\begin{center}
	\begin{tabular}{|c|ccc|}
		\hline
		$\underline{\circ_1} $ & 1 & u & 0 \\
		\hline
		1 & 1 & 1 & u \\
		u & 1 & 1 & 1 \\
		0 & 1 & 1 & u \\
		\hline
	\end{tabular}
	\begin{tabular}{|c|ccc|}
		\hline
		$\underline{\circ_2} $ & 1 & u & 0 \\
		\hline
		1 & 1 & 1 & u \\
		u & 1 & 1 & 1 \\
		0 & 1 & 1 & 1 \\
		\hline
	\end{tabular}
\end{center}

They do not have any corresponding expressions in natural language, but their importance lies in a different field: as follows from Pa\l{}asi\'{n}ska's research \cite{Palasinska}, the 1-matrices with $ \circ_1 $ or $\circ_2$ as their only connective are not finitely axiomatizable. Pa\l{}asi\'{n}ska had in mind rules of the following shape (notice that what she calls \textit{terms} we call \textit{formulae}): 
``By a \textit{standard inference rule}, or simply a \textit{rule}, we mean a pair $ \langle X,t\rangle$, where $ X $ is a finite set of terms and $ t $ is a term.'' \cite[p. 362]{Palasinska}. However, we are not dealing with standard inference rules but with bisequent rules. It appears that  what is impossible in the standard framework becomes possible in our bisequent approach, which shows that it is indeed a uniform and general method.

To characterise the first operator we need three rules:

\begin{flushleft}
	$(\mathord{\circ_1}\mathord{\Rightarrow}\mid)$ \, $\dfrac{\Gamma
		\Rightarrow \Delta, \varphi \mid \varphi, \Pi\Rightarrow\Sigma \qquad \Gamma\Rightarrow\Delta \mid\psi,
		\Pi \Rightarrow \Sigma } {\varphi\circ_1\psi, \Gamma
		\Rightarrow \Delta \mid \Pi\Rightarrow\Sigma}$
\end{flushleft} 

\begin{flushleft}
	$(\mathord{\Rightarrow}\mathord{\circ_1}\mid)$ \, $\dfrac{\varphi, \Gamma
		\Rightarrow \Delta \mid \Pi\Rightarrow\Sigma, \psi \qquad\Gamma\Rightarrow\Delta \mid\Pi \Rightarrow \Sigma, \varphi, \psi} {\Gamma
		\Rightarrow \Delta, \varphi\circ_1\psi \mid \Pi\Rightarrow\Sigma}$
\end{flushleft} 

\begin{flushleft}
	$(\mathord{\mid\mathord{\circ_1}\Rightarrow})$ \, $\dfrac{\varphi, \Gamma
		\Rightarrow \Delta \mid \Pi\Rightarrow\Sigma \qquad\Gamma\Rightarrow\Delta, \varphi \mid\Pi \Rightarrow \Sigma} {\Gamma
		\Rightarrow \Delta\mid \varphi\circ_1\psi, \Pi\Rightarrow\Sigma}$
\end{flushleft}

and the additional axiom schema: $\Gamma\Rightarrow\Delta\mid\Pi\Rightarrow\Sigma, \varphi\circ_1\psi$

\1

Note that we need this axiom, instead of the rule, since the formula $\varphi\circ_1\psi$ is never false. It is a general principle that if we characterise in BSC an operation with no output 0 we must add such kind of axiom schema. Dually, in case of the operation which never gives an output 1, we have to add an axiom schema with an arbitrary formula made by this operation in the antecedent of a 1-sequent.  Otherwise the formalisation of this operation is incomplete.

Note also that the rule $(\mathord{\mid\mathord{\circ_1}\Rightarrow})$ is somewhat arbitrary. To obtain the same effect we could use instead:

\begin{flushleft}
	$(\mathord{\mid\mathord{\circ_1}\Rightarrow})$ \, $\dfrac{\Gamma
		\Rightarrow \Delta \mid \varphi, \Pi\Rightarrow\Sigma \qquad\Gamma\Rightarrow\Delta\mid\Pi \Rightarrow \Sigma, \varphi} {\Gamma
		\Rightarrow \Delta\mid \varphi\circ_1\psi, \Pi\Rightarrow\Sigma}$
\end{flushleft} 

or one of these two rules but with $\psi$ instead of $\varphi$ as a side formula. Each variant covers all positions in the table which agree with the fact that the principal formula is never false, hence always true or undefined. The same effect can be obtained in totally different (and significantly more complex) way by using four-premiss rules with specific combinations of $\varphi, \psi$ in each premiss. We omit the details but it should be mentioned to emphasize the sense of our claim that the construction of rules on the basis of tabular representation of operations is not deterministic. The more bizarre operations we consider, the more options are possible, but simplicity should be always the important factor in choosing one of them. BSC characterisation of $\circ_2$ needs only simplification of the first two rules for $\circ_1$; in $(\mathord{\circ_2}\mathord{\Rightarrow}\mid)$ we have to delete the second occurrence of $\varphi$ in the first premiss, in $(\mathord{\Rightarrow}\mathord{\circ_2}\mid)$ we have to delete the second premiss. The remaining rule and the axiom are similar to those for $\circ_1$.

\1

We would like to finish this section with the consideration of
logics determined by different notions of consequence relations. Two relations considered in the text express informally the situation where either truth is preserved or non-falsity is preserved. But two other possibilities are open as well: $\Gamma\Rightarrow\,\mid\, \Rightarrow\psi$ corresponds to the notion of no-counterexample consequence (see e.g. Lehmann \cite{Lehmann}, Paoli \cite{Paoli}), whereas $\Rightarrow \varphi \mid  \Gamma\Rightarrow$ corresponds to the liberal consequence which leads from non-falsity to truth. So the only thing we have to change is the definition of proof in the part concerned with the root-sequent. 
It is an additional advantage of BSC that having rules established for the operations in the context of standard notion of consequence, we need only to change slightly the definition of proof while still keeping the rules and axioms intact.

\section{Syntactic Proof of Admissibility of Cut and other Structural Rules}\label{Syntactic}

The most suitable strategy for providing a constructive proof of cut admissibility for BSC is the one applied for the standard sequent calculus G3 in Negri and von Plato \cite{neg:str01}. Briefly, we must prove three things: (a) that atomic axiomatic sequents are sufficient for proving general cases; (b) that the admissibility of weakening, (syntactical) invertibility of all rules, and admissibility of contraction hold; and (c) that on this basis the simultaneous admissibility of both cut rules is provable by double induction on the height of proofs of their premisses and the complexity of cut-formula. Proving (a) in this context is necessary not only for showing the height-preserving invertibility of rules but also to simplify the proof of (c). 
The proof of part (b) is standard; the only specific and difficult cases are in proving some subcases of (a) and in proving some subcases of (c) by induction on the complexity of cut-formulae. Let us now present this proof in details.

\begin{lemma}\label{Axiomatic}
	Axiomatic sequents may be replaced with atomic axioms.
\end{lemma}
\begin{proof}
	By induction on the complexity of active formula. As an example, let us prove this lemma for \L ukasiewicz's implication. We have three cases, by the number of types of axioms.  Consider the first case:
\begin{prooftree}
\EnableBpAbbreviations
\AXC{$ \varphi,\Gamma
	\Rightarrow \Delta, \varphi \mid \Pi\Rightarrow\Sigma,\psi $}	
\AXC{$ \varphi,\Gamma\Rightarrow\Delta \mid \psi,
	\Pi \Rightarrow \Sigma,\psi $}
\RL{$(\mathord{\mid\rightarrow}\mathord{\Rightarrow})$}
\BIC{$ \varphi,\Gamma
	\Rightarrow \Delta \mid \varphi\rightarrow\psi, \Pi\Rightarrow\Sigma,\psi $}
\RL{$(\mathord{\mid\Rightarrow}\mathord{\rightarrow})$}
\UIC{$ \Gamma
	\Rightarrow \Delta \mid \varphi\rightarrow\psi, \Pi\Rightarrow\Sigma,\varphi\rightarrow\psi $}
\end{prooftree}

Consider the second case. We apply  $(\mathord{\rightarrow_L}\mathord{\Rightarrow}\mid)$ twice as follows.

\begin{center}
{\footnotesize 	$ \dfrac{\varphi,\Gamma
		\Rightarrow \Delta, \varphi,\psi \mid \psi, \Pi\Rightarrow\Sigma \qquad \varphi, \psi, \Gamma\Rightarrow\Delta,\psi \mid
		\Pi \Rightarrow \Sigma \qquad \varphi,\Gamma\Rightarrow\Delta,\psi \mid \Pi\Rightarrow\Sigma, \varphi}{\varphi,\varphi\rightarrow\psi, \Gamma
		\Rightarrow \Delta,\psi \mid \Pi\Rightarrow\Sigma} $}
\end{center}
\begin{center}
	{\footnotesize 	$ \dfrac{\Gamma
			\Rightarrow \Delta, \varphi \mid \varphi,\psi, \Pi\Rightarrow\Sigma,\psi \qquad  \psi, \Gamma\Rightarrow\Delta \mid
			\varphi,\Pi \Rightarrow \Sigma,\psi \qquad \Gamma\Rightarrow\Delta \mid \varphi,\Pi\Rightarrow\Sigma, \varphi,\psi}{\varphi\rightarrow\psi, \Gamma
			\Rightarrow \Delta \mid \varphi,\Pi\Rightarrow\Sigma,\psi} $}
\end{center}

Now we apply $(\mathord{\Rightarrow}\mathord{\rightarrow_L}\mid)$ as follows.

\begin{center}
	$ \dfrac{\varphi,\varphi\rightarrow\psi, \Gamma
		\Rightarrow \Delta,\psi \mid \Pi\Rightarrow\Sigma \qquad \varphi\rightarrow\psi, \Gamma
		\Rightarrow \Delta \mid \varphi,\Pi\Rightarrow\Sigma,\psi}{\varphi\rightarrow\psi, \Gamma
		\Rightarrow \Delta,\varphi\rightarrow\psi \mid \Pi\Rightarrow\Sigma} $
\end{center}	

Consider the third case. We apply the rules $(\mathord{\rightarrow_L}\mathord{\Rightarrow}\mid)$ and $(\mathord{\mid\Rightarrow}\mathord{\rightarrow})$:
{\footnotesize \begin{prooftree}
	\EnableBpAbbreviations
\AXC{$\varphi,\Gamma
	\Rightarrow \Delta, \varphi \mid \psi, \Pi\Rightarrow\Sigma,\psi \qquad \varphi, \psi, \Gamma\Rightarrow\Delta \mid
	\Pi \Rightarrow \Sigma,\psi \qquad \varphi,\Gamma\Rightarrow\Delta \mid \Pi\Rightarrow\Sigma, \varphi,\psi $}
\UIC{$ \varphi,\varphi\rightarrow\psi, \Gamma
	\Rightarrow \Delta \mid \Pi\Rightarrow\Sigma,\psi $}
\UIC{$ \varphi\rightarrow\psi, \Gamma
	\Rightarrow \Delta \mid \Pi\Rightarrow\Sigma,\varphi\rightarrow\psi $}
\end{prooftree}}
	
\end{proof}

Lemma \ref{Axiomatic} allows us to consider in the following only proofs with atomic axioms.

We write $ \vdash_n\Gamma
\Rightarrow \Delta \mid \Pi\Rightarrow\Sigma$  iff a proof of $ \Gamma
\Rightarrow \Delta \mid \Pi\Rightarrow\Sigma $ is of height at most $ n $. We write $ \vdash_{<n}\Gamma
\Rightarrow \Delta \mid \Pi\Rightarrow\Sigma$  iff $ \Gamma
\Rightarrow \Delta \mid \Pi\Rightarrow\Sigma $ is provable with a proof of height $ < n $. 
Consider the following versions of weakening which can be formulated in a bisequent framework. 

\1

\noindent\begin{tabular}{ll}	
	$(\mathord{W}\mathord{\Rightarrow\mid})$ \ $\dfrac{\Gamma
		\Rightarrow \Delta \mid S}{\varphi, \Gamma \Rightarrow \Delta \mid S}$ &
	$(\mathord{\Rightarrow}\mathord{W\mid})$ \ $\dfrac{\Gamma
		\Rightarrow \Delta \mid S}{\Gamma
		\Rightarrow \Delta, \varphi \mid S}$\\[16pt]
	
	$(\mathord{\mid W}\mathord{\Rightarrow})$ \ $\dfrac{S \mid \Pi\Rightarrow\Sigma}{S\mid \varphi, \Pi\Rightarrow\Sigma}$ &
	$(\mathord{\mid\Rightarrow}\mathord{W})$ \ $\dfrac{S \mid \Pi\Rightarrow\Sigma}{S \mid \Pi\Rightarrow\Sigma, \varphi}$\\[16pt]	
\end{tabular}
\begin{theorem}[Height-preserving admissibility of Weakening in BSC-L]\label{Weakening}
Let $ \bf L$ be one of the logics considered in the paper. If BSK-L $ \vdash_n \Gamma\Rightarrow\Delta\mid\Pi\Rightarrow\Sigma$, then BSK-L $ \vdash_n \Gamma^\prime\Rightarrow\Delta^\prime\mid\Pi^\prime\Rightarrow\Sigma^\prime$, where $ \Gamma\subseteq\Gamma^\prime $, $ \Delta\subseteq\Delta^\prime $, $ \Pi\subseteq\Pi^\prime $, and $ \Sigma\subseteq\Sigma^\prime $.	
\end{theorem}
\begin{proof}
	By induction on the height of the premiss. In the basis, if $ \Gamma\Rightarrow\Delta\mid\Pi\Rightarrow\Sigma $ is
	axiomatic, then $ \vdash_0\Gamma\Rightarrow\Delta\mid\Pi\Rightarrow\Sigma $ implies $ \vdash_0 \Gamma^\prime\Rightarrow\Delta^\prime\mid\Pi^\prime\Rightarrow\Sigma^\prime$. In the
	inductive step, we have to consider the cases generated by the logical rules. We examine some examples.
	
Case $ (\Rightarrow\rightarrow\mid) $; $ \Gamma\Rightarrow\Delta\mid\Pi\Rightarrow\Sigma:=\Gamma\Rightarrow\Theta,\varphi\rightarrow\psi\mid\Pi\Rightarrow\Sigma $. The premiss is $ \Gamma\Rightarrow\Theta,\psi\mid\varphi,\Pi\Rightarrow\Sigma $ with a proof having height $ < n $. By the induction hypothesis, BSC-L $ \vdash_{<n}\Gamma^\prime\Rightarrow\Theta^\prime,\psi\mid\varphi,\Pi^\prime\Rightarrow\Sigma^\prime $, where $ \Theta^\prime=\Delta^\prime\setminus\{\varphi\rightarrow\psi\} $. Hence, by $ (\Rightarrow\rightarrow\mid) $, BSC-L $ \vdash_{n}\Gamma^\prime\Rightarrow\Delta^\prime\mid\Pi^\prime\Rightarrow\Sigma^\prime $.

Case $ (\mid\rightarrow\Rightarrow) $; $ \Gamma\Rightarrow\Delta\mid\Pi\Rightarrow\Sigma:=\Gamma\Rightarrow\Delta\mid\varphi\rightarrow\psi,\Theta\Rightarrow\Sigma$. As for the premisses, we have BSC-L $ \vdash_{<n}\Gamma\Rightarrow\Delta,\varphi\mid\Theta\Rightarrow\Sigma$ as well as BSC-L  $ \vdash_{<n}\Gamma\Rightarrow\Delta\mid\psi,\Theta\Rightarrow\Sigma$. By the induction hypothesis, it holds that BSC-L $ \vdash_{<n}\Gamma^\prime\Rightarrow\Delta^\prime,\varphi\mid\Theta^\prime\Rightarrow\Sigma^\prime$ and BSC-L  $ \vdash_{<n}\Gamma^\prime\Rightarrow\Delta^\prime\mid\psi,\Theta^\prime\Rightarrow\Sigma^\prime$, where $ \Theta^\prime=\Pi^\prime\setminus\{\varphi\rightarrow\psi\} $. Thus, by $ (\mid\rightarrow\Rightarrow) $, BSC-L $ \vdash_{n}\Gamma^\prime\Rightarrow\Delta^\prime\mid\Pi^\prime\Rightarrow\Sigma^\prime $.
\end{proof}

\begin{theorem}[Height-preserving invertibility]\label{Invertability}
For any instance of the application
of any logical rule, if the conclusion has a proof of height $ n $, then the premisses have
proofs of height $\leqslant n $.	
\end{theorem}
\begin{proof}
By induction on the height.	As an example, consider the case $ (\Rightarrow\rightarrow\mid) $. In the basis, we deal with an axiom $ p,\Gamma\Rightarrow\Delta,p,\varphi\rightarrow\psi\mid\Pi\Rightarrow\Sigma $. Then $ p,\Gamma\Rightarrow\Delta,p,\psi\mid\varphi,\Pi\Rightarrow\Sigma $ is an axiom as well. Consider the induction step. We have  $\vdash_n\Gamma\Rightarrow\Delta,\varphi\rightarrow\psi\mid\Pi\Rightarrow\Sigma $. There are two subcases depending on whether the formula $ \varphi\rightarrow\psi $ is principal. Suppose that it is principal. Then $\vdash_{n-1}\Gamma\Rightarrow\Delta,\psi\mid\varphi,\Pi\Rightarrow\Sigma $. Suppose that it is not principal. Then we have to consider the last rule application. Since all our rules are context independent, we may demonstrate it schematically. Suppose that the last rule is $ (\neg \Rightarrow\mid) $. Then we have: 
\begin{center}
	$ \dfrac{\Gamma\Rightarrow\Delta,\varphi\rightarrow\psi\mid\Pi\Rightarrow\Sigma,\omega}{\neg\omega,\Gamma\Rightarrow\Delta,\varphi\rightarrow\psi\mid\Pi\Rightarrow\Sigma} $
\end{center}

By the inductive hypothesis, $ \Gamma\Rightarrow\Delta,\psi\mid\varphi,\Pi\Rightarrow\Sigma,\omega $ is provable with lower height. By the same rule we get $ \neg\omega,\Gamma\Rightarrow\Delta,\psi\mid\varphi,\Pi\Rightarrow\Sigma $.

Suppose that the last rule is $(\mid\neg_P\Rightarrow)$. Then we have: 
\begin{center}
	$ \dfrac{\Gamma\Rightarrow\Delta,\varphi\rightarrow\psi\mid\Pi\Rightarrow\Sigma,\omega\qquad\omega,\Gamma\Rightarrow\Delta,\varphi\rightarrow\psi\mid\Pi\Rightarrow\Sigma}{\Gamma\Rightarrow\Delta,\varphi\rightarrow\psi\mid\neg\omega,\Pi\Rightarrow\Sigma} $
\end{center}

By the inductive hypothesis, $ \Gamma\Rightarrow\Delta,\psi\mid\varphi,\Pi\Rightarrow\Sigma,\omega $ and $ \omega,\Gamma\Rightarrow\Delta,\psi\mid\varphi,\Pi\Rightarrow\Sigma $ are provable with lower height. By the same rule we get $ \Gamma\Rightarrow\Delta,\psi\mid\varphi,\neg\omega,\Pi\Rightarrow\Sigma $. The other cases are proved similarly.
\end{proof}

Consider the following versions of contraction which can be formulated in a bisequent framework. 

\noindent\begin{tabular}{ll}	
	$(\mathord{C}\mathord{\Rightarrow\mid})$ \ $\dfrac{\varphi,\varphi,\Gamma
		\Rightarrow \Delta \mid S}{\varphi, \Gamma \Rightarrow \Delta \mid S}$ &
	$(\mathord{\Rightarrow}\mathord{C\mid})$ \ $\dfrac{\Gamma
		\Rightarrow \Delta, \varphi, \varphi \mid S}{\Gamma
		\Rightarrow \Delta, \varphi \mid S}$\\[16pt]
	
	$(\mathord{\mid C}\mathord{\Rightarrow})$ \ $\dfrac{S \mid\varphi,\varphi, \Pi\Rightarrow\Sigma}{S\mid \varphi, \Pi\Rightarrow\Sigma}$ &
	$(\mathord{\mid\Rightarrow}\mathord{C})$ \ $\dfrac{S \mid \Pi\Rightarrow\Sigma, \varphi, \varphi}{S \mid \Pi\Rightarrow\Sigma, \varphi}$\\[16pt]	
\end{tabular}
\begin{theorem}[Height-preserving admissibility of contraction]\label{Contraction}
	$  $
\begin{itemize}\itemsep=0pt
\item	BSC-L $ \vdash_n\varphi,\varphi,\Gamma
\Rightarrow \Delta \mid S $ implies BSC-L $ \vdash_n\varphi,\Gamma
\Rightarrow \Delta \mid S $;
\item BSC-L $ \vdash_n\Gamma
\Rightarrow \Delta,\varphi,\varphi \mid S $ implies BSC-L $ \vdash_n\Gamma
\Rightarrow \Delta,\varphi \mid S $;
\item BSC-L $\vdash_nS \mid\varphi,\varphi, \Pi\Rightarrow\Sigma $ implies BSC-L $\vdash_nS \mid\varphi, \Pi\Rightarrow\Sigma $;
\item BSC-L $\vdash_nS \mid \Pi\Rightarrow\Sigma,\varphi,\varphi $ implies BSC-L $\vdash_nS \mid\Pi\Rightarrow\Sigma,\varphi $.
\end{itemize}
\end{theorem}
\begin{proof}
By induction on the height of proof of $ \varphi,\varphi,\Gamma
\Rightarrow \Delta \mid S $	(the other cases are considered similarly). In the basis, we have that $ \vdash_0\varphi,\varphi,\Gamma
\Rightarrow \Delta \mid S $ and that it is an axiom. Hence, $ \vdash_0\varphi,\Gamma
\Rightarrow \Delta \mid S $ and it is an axiom as well.

As for the induction step, we need to distinguish the cases depending on whether $ \varphi $ is principal. Suppose that $ \varphi $ is not principal. Then $ \varphi,\varphi,\Gamma
\Rightarrow \Delta \mid S $ is derived from $ \varphi,\varphi,\Gamma^\prime
\Rightarrow \Delta^\prime \mid S^\prime $ by a one-premiss
rule or additionally from $ \varphi,\varphi,\Gamma^{\prime\prime}
\Rightarrow \Delta^{\prime\prime} \mid S^{\prime\prime} $ by a two-premiss rule or also from $ \varphi,\varphi,\Gamma^{\prime\prime\prime}
\Rightarrow \Delta^{\prime\prime\prime} \mid S^{\prime\prime\prime} $ by a three-premiss rule. By the induction hypothesis, $\vdash_{<n} \varphi,\Gamma^\prime
\Rightarrow \Delta^\prime \mid S^\prime $, $\vdash_{<n} \varphi,\Gamma^{\prime\prime}
\Rightarrow \Delta^{\prime\prime} \mid S^{\prime\prime} $ and $\vdash_{<n} \varphi,\Gamma^{\prime\prime\prime}
\Rightarrow \Delta^{\prime\prime\prime} \mid S^{\prime\prime\prime} $. The application of the same rules yields $\vdash_{<n} \varphi,\Gamma
\Rightarrow \Delta \mid S $.

Suppose that $ \varphi $ is principal. Assume that $ \varphi $ is $ \chi\vee\omega $ and is obtained by the rule $ (\vee_w\Rightarrow\mid) $. Then the sequent in question has the form $$ \chi\vee\omega,\chi\vee\omega,\Gamma\Rightarrow\Delta\mid \Pi\Rightarrow\Sigma $$ and is obtained from the following premisses, each of them having a proof of the height $<n$: 
\begin{enumerate}\itemsep=0pt
\item $  \chi,\omega,\chi\vee\omega,\Gamma\Rightarrow\Delta\mid \Pi\Rightarrow\Sigma $
\item $ \chi,\chi\vee\omega,\Gamma\Rightarrow\Delta\mid \Pi\Rightarrow\Sigma,\omega $
\item $ \omega,\chi\vee\omega,\Gamma\Rightarrow\Delta\mid \Pi\Rightarrow\Sigma,\chi $.
\end{enumerate}

By Theorem \ref{Invertability}, \begin{enumerate}\itemsep=0pt
\item	\begin{enumerate}\itemsep=0pt
\item	$\vdash_{<n-1}  \chi,\omega,\chi,\omega,\Gamma\Rightarrow\Delta\mid \Pi\Rightarrow\Sigma $,
\item $\vdash_{<n-1}  \chi,\omega,\chi,\Gamma\Rightarrow\Delta\mid \Pi\Rightarrow\Sigma,\omega $, 
\item $\vdash_{<n-1}  \chi,\omega,\omega,\Gamma\Rightarrow\Delta\mid \Pi\Rightarrow\Sigma,\chi $;
\end{enumerate}
\item \begin{enumerate}\itemsep=0pt
\item $\vdash_{<n-1} \chi,\chi,\omega,\Gamma\Rightarrow\Delta\mid \Pi\Rightarrow\Sigma,\omega $,
\item $\vdash_{<n-1} \chi,\chi,\Gamma\Rightarrow\Delta\mid \Pi\Rightarrow\Sigma,\omega,\omega $,
\item $\vdash_{<n-1} \chi,\omega,\Gamma\Rightarrow\Delta\mid \Pi\Rightarrow\Sigma,\omega,\chi $;
\end{enumerate}
\item \begin{enumerate}\itemsep=0pt
	\item $\vdash_{<n-1} \omega,\chi,\omega,\Gamma\Rightarrow\Delta\mid \Pi\Rightarrow\Sigma,\chi $,
	\item $\vdash_{<n-1} \omega,\chi,\Gamma\Rightarrow\Delta\mid \Pi\Rightarrow\Sigma,\chi,\omega $,
	\item $\vdash_{<n-1} \omega,\omega,\Gamma\Rightarrow\Delta\mid \Pi\Rightarrow\Sigma,\chi,\chi $.	
\end{enumerate}
\end{enumerate}

Consider the sequents 1. (a), 2. (b), and 3. (c). We apply the induction hypothesis to each of them twice and get 
\begin{itemize}\itemsep=0pt
\item $\vdash_{<n-1}  \chi,\omega,\Gamma\Rightarrow\Delta\mid \Pi\Rightarrow\Sigma $,
\item $\vdash_{<n-1} \chi,\Gamma\Rightarrow\Delta\mid \Pi\Rightarrow\Sigma,\omega $,
	\item $\vdash_{<n-1} \omega,\Gamma\Rightarrow\Delta\mid \Pi\Rightarrow\Sigma,\chi$.
\end{itemize}

Now it is time to apply the rule $ (\vee_w\Rightarrow\mid) $. As a result, we obtain that $\vdash_{n}  \chi\vee\omega,\Gamma\Rightarrow\Delta\mid \Pi\Rightarrow\Sigma $. The other cases are handled similarly.
\end{proof}

In a bisequent framework there are two versions of the cut rule  presented on page \pageref{Cut}\footnote{In fact one may consider also other forms of cut, see e.g. \cite{Fjellstad1}, but these can be easily shown to be derivable in the presence of the two variants of cut we consider.}. We prove the admissibility of both of them by Dragalin's method \cite{Dragalin} (see also  \cite{Indrzejczak2,TroelstraSchwichtenberg,neg:str01} for the presentation of this method; in particular, in \cite{Indrzejczak2}, this method is applied to structured SC which formalize some many-valued logics). 
\begin{theorem}[Admissibility of $ (Cut\mid) $ and $ (\mid Cut) $]
If BSC-L $ \vdash \Gamma^\prime\Rightarrow\Delta^\prime,\varphi\mid \Lambda^\prime\Rightarrow\Theta^\prime$ and BSC-L $ \vdash \varphi,\Pi^\prime\Rightarrow\Sigma^\prime \mid \Xi^\prime\Rightarrow\Omega^\prime$, then BSC-L $ \vdash \Gamma^\prime,\Pi^\prime\Rightarrow\Delta^\prime,\Sigma^\prime\mid \Lambda^\prime,\Xi^\prime\Rightarrow\Theta^\prime,\Omega^\prime$. If BSC-L $ \vdash \Gamma^\prime\Rightarrow\Delta^\prime\mid \Lambda^\prime\Rightarrow\Theta^\prime,\varphi$ and BSC-L $ \vdash \Pi^\prime\Rightarrow\Sigma^\prime \mid \varphi,\Xi^\prime\Rightarrow\Omega^\prime$, then BSC-L $ \vdash \Gamma^\prime,\Pi^\prime\Rightarrow\Delta^\prime,\Sigma^\prime\mid \Lambda^\prime,\Xi^\prime\Rightarrow\Theta^\prime,\Omega^\prime$.
\end{theorem}
\begin{proof} We carry the proof for both forms of cut simultaneously. The proof is by 
	double induction on the complexity of the cut-formula and on the
	sum of heights of proofs of both premisses. In the induction hypotheses we assume that the claim holds for
	both forms of cut. The proof is divided into three parts: 
\begin{enumerate}\itemsep=0pt
\item at least one premiss is axiomatic;
\item the cut-formula is principal in both premisses;
\item the cut-formula is not principal in at least one premiss.
\end{enumerate}

In cases 2 and 3 we provide only some characteristic examples for the sake of illustration.

\1

\textbf{Case 1.} Let the left premiss of the application of $ (Cut\mid) $ be axiomatic. Since there are three types of axioms we have three subcases. 

\textit{Case 1.1.} The cut-formula is active. \textit{Case 1.1.1.} We have the subsequent application of  $ (Cut\mid) $:
\begin{center}
	$\dfrac{\varphi,\Gamma 
	\Rightarrow  \Delta, \varphi \mid \Lambda\Rightarrow\Theta \qquad
	\varphi, \Pi  \Rightarrow  \Sigma \mid \Xi\Rightarrow\Omega} {\varphi,\Gamma, \Pi  \Rightarrow  \Delta, \Sigma \mid \Lambda, \Xi\Rightarrow\Theta, \Omega}$
\end{center}

By Theorem \ref{Weakening}, we obtain that since $ \varphi, \Pi  \Rightarrow  \Sigma \mid \Xi\Rightarrow\Omega $ is provable, 
 $ \varphi,\Gamma, \Pi  \Rightarrow  \Delta, \Sigma \mid \Lambda, \Xi\Rightarrow\Theta, \Omega $ is provable as well (without using $ (Cut\mid) $). In the case of $ (\mid Cut) $, the proof is analogous. The \textit{Case 1.1.2.} with $\Gamma 
 \Rightarrow  \Delta \mid \varphi, \Lambda\Rightarrow\Theta, \varphi$ is analogous. 

\textit{Case 1.1.3.} We have the subsequent application of  $ (\mid Cut) $:
\begin{center}
	$\dfrac{\varphi,\Gamma 
	\Rightarrow  \Delta \mid \Lambda\Rightarrow\Theta, \varphi \qquad
	 \Pi  \Rightarrow  \Sigma \mid \varphi,\Xi\Rightarrow\Omega} {\varphi,\Gamma, \Pi  \Rightarrow  \Delta, \Sigma \mid \Lambda, \Xi\Rightarrow\Theta, \Omega}$
\end{center}

In this case we need an additional induction on the height of the right premiss. In the basis there are two subcases. The first one: $ \Pi  \Rightarrow  \Sigma \mid \varphi,\Xi\Rightarrow\Omega $ is an axiom with $ \varphi $ being parametric. Then $ \varphi,\Gamma, \Pi  \Rightarrow  \Delta, \Sigma \mid \Lambda, \Xi\Rightarrow\Theta, \Omega $ is an axiom as well. The second one: $ \varphi\in\Omega $. Again we obtain that $ \varphi,\Gamma, \Pi  \Rightarrow  \Delta, \Sigma \mid \Lambda, \Xi\Rightarrow\Theta, \Omega $ is an axiom as well. As for the inductive step, note that $\varphi$ is atomic, due to lemma \ref{Axiomatic}, hence it cannot be a principal formula of $\Pi  \Rightarrow  \Sigma \mid \varphi,\Xi\Rightarrow\Omega$ obtained by some rule $r$. So we obtain the result by the induction hypothesis, $r$, and contractions, if necessary. The subcase with $ (Cut \mid) $ is dealt with in the same way.

\textit{Case 1.2.} the cut-formula is parametric. As an example, consider the case of $ (Cut\mid) $. There are three possibilities:
\begin{center}
	$\dfrac{\psi,\Gamma 
		\Rightarrow  \Delta,\psi, \varphi \mid \Lambda\Rightarrow\Theta \qquad
		\varphi, \Pi  \Rightarrow  \Sigma \mid \Xi\Rightarrow\Omega} {\psi,\Gamma, \Pi  \Rightarrow  \Delta,\psi, \Sigma \mid \Lambda, \Xi\Rightarrow\Theta, \Omega}$
\end{center}
\begin{center}
	$\dfrac{\psi,\Gamma 
		\Rightarrow  \Delta, \varphi \mid \Lambda\Rightarrow\Theta,\psi \qquad
		\varphi, \Pi  \Rightarrow  \Sigma \mid \Xi\Rightarrow\Omega} {\psi,\Gamma, \Pi  \Rightarrow  \Delta, \Sigma \mid \Lambda, \Xi\Rightarrow\Theta,\psi, \Omega}$
\end{center}
\begin{center}
	$\dfrac{\Gamma 
		\Rightarrow  \Delta, \varphi \mid \psi,\Lambda\Rightarrow\Theta,\psi \qquad
		\varphi, \Pi  \Rightarrow  \Sigma \mid \Xi\Rightarrow\Omega} {\Gamma, \Pi  \Rightarrow  \Delta, \Sigma \mid \psi,\Lambda, \Xi\Rightarrow\Theta,\psi, \Omega}$
\end{center}

In all these possibilities, the conclusion is already provable as
an axiom. If the right premiss is axiomatic, then the situation is analogous and is considered similarly.

\textbf{Case 2.} The cut-formula is principal in both premisses. 

\textit{Case 2.1.} The cut-formula is Kleene's negation.
\begin{prooftree}
	\EnableBpAbbreviations
	\AXC{$\Gamma\Rightarrow\Delta\mid\varphi,\Pi\Rightarrow\Sigma$}
	\RL{$ (\Rightarrow\neg\mid) $}
	\UIC{$\Gamma\Rightarrow\Delta,\neg\varphi\mid\Pi\Rightarrow\Sigma$}
	\AXC{$\Theta\Rightarrow\Lambda\mid\Xi\Rightarrow\Omega,\varphi$}
	\RL{$ (\neg\Rightarrow\mid) $}
	\UIC{$\neg\varphi,\Theta\Rightarrow\Lambda\mid\Xi\Rightarrow\Omega$}
	\RL{$ (Cut\mid) $}
	\BIC{$\Gamma,\Theta\Rightarrow\Delta,\Lambda\mid\Pi,\Xi\Rightarrow\Sigma,\Omega$}	
\end{prooftree}

We replace it with the following derivation, applying another version of cut, $ (\mid Cut) $, on the formula with lower complexity:
\begin{prooftree}
	\EnableBpAbbreviations
	\AXC{$\Theta\Rightarrow\Lambda\mid\Xi\Rightarrow\Omega,\varphi$}
	\AXC{$\Gamma\Rightarrow\Delta\mid\varphi,\Pi\Rightarrow\Sigma$}
	\RL{$ (\mid Cut) $}
	\BIC{$\Gamma,\Theta\Rightarrow\Delta,\Lambda\mid\Pi,\Xi\Rightarrow\Sigma,\Omega$}	
\end{prooftree}

Kleene's negation again, but this time with the application of $ (\mid Cut) $.
\begin{prooftree}
	\EnableBpAbbreviations
	\AXC{$\varphi,\Gamma\Rightarrow\Delta\mid\Pi\Rightarrow\Sigma$}
	\RL{$ (\mid\Rightarrow\neg) $}
	\UIC{$\Gamma\Rightarrow\Delta\mid\Pi\Rightarrow\Sigma,\neg\varphi$}
	\AXC{$\Theta\Rightarrow\Lambda,\varphi\mid\Xi\Rightarrow\Omega$}
	\RL{$ (\mid\neg\Rightarrow) $}
	\UIC{$\Theta\Rightarrow\Lambda\mid\neg\varphi,\Xi\Rightarrow\Omega$}
	\RL{$ (\mid Cut) $}
	\BIC{$\Gamma,\Theta\Rightarrow\Delta,\Lambda\mid\Pi,\Xi\Rightarrow\Sigma,\Omega$}	
\end{prooftree}

We replace it with:
\begin{prooftree}
	\EnableBpAbbreviations
	\AXC{$\Theta\Rightarrow\Lambda,\varphi\mid\Xi\Rightarrow\Omega$}
	\AXC{$\varphi,\Gamma\Rightarrow\Delta\mid\Pi\Rightarrow\Sigma$}
	\RL{$ (Cut\mid) $}
	\BIC{$\Gamma,\Theta\Rightarrow\Delta,\Lambda\mid\Pi,\Xi\Rightarrow\Sigma,\Omega$}	
\end{prooftree}

\textit{Case 2.2.} The cut-formula is \L ukasiewicz's implication. We apply $ (Cut\mid) $ as follows:
\begin{center}
	$ \dfrac{\Gamma\Rightarrow\Delta, \varphi\rightarrow\psi \mid \Pi\Rightarrow\Sigma \qquad \varphi\rightarrow\psi,\Theta\Rightarrow\Lambda \mid \Xi\Rightarrow\Omega}{\Gamma,\Theta\Rightarrow\Delta,\Lambda\mid\Pi,\Xi\Rightarrow\Sigma,\Omega} $
\end{center}
where the premisses are obtained by the rules $(\mathord{\Rightarrow}\mathord{\rightarrow_L}\mid)$ and $(\mathord{\rightarrow_L}\mathord{\Rightarrow}\mid)$, respectively:
\begin{center}
	$\dfrac{\varphi, \Gamma
		\Rightarrow \Delta, \psi \mid \Pi\Rightarrow\Sigma \qquad \Gamma\Rightarrow\Delta \mid \varphi, \Pi\Rightarrow\Sigma, \psi}{\Gamma\Rightarrow\Delta, \varphi\rightarrow\psi \mid \Pi\Rightarrow\Sigma}$
\end{center}
\begin{center}
	{\small 	$\dfrac{\Theta
			\Rightarrow \Lambda, \varphi \mid \psi, \Xi\Rightarrow\Omega \quad \psi, \Theta\Rightarrow\Lambda \mid
			\Xi \Rightarrow \Omega \quad \Theta\Rightarrow\Lambda \mid \Xi\Rightarrow\Omega, \varphi } {\varphi\rightarrow\psi, \Theta
			\Rightarrow \Lambda \mid \Xi\Rightarrow\Omega}$}
\end{center}

We transform this deduction as follows, applying $ (\mid Cut) $ and $ (Cut\mid ) $ to the formulae with lower complexity as well as several times using contraction rules (see Theorem \ref{Contraction}) which we indicate by double lines. We start with the following deduction:
\begin{prooftree}
	\EnableBpAbbreviations
	\AXC{$ \Theta\Rightarrow\Lambda, \varphi \mid \psi,\Xi\Rightarrow\Omega $}
	\AXC{$\varphi, \Gamma
		\Rightarrow \Delta, \psi \mid \Pi\Rightarrow\Sigma$}
	\AXC{$\psi, \Theta\Rightarrow\Lambda \mid
		\Xi \Rightarrow \Omega$}
	\BIC{$\varphi, \Gamma,\Theta
		\Rightarrow \Delta,\Lambda  \mid \Pi,\Xi\Rightarrow\Sigma,\Omega$}
	\BIC{$\Gamma,\Theta,\Theta
		\Rightarrow \Delta,\Lambda,\Lambda  \mid \psi,\Pi,\Xi,\Xi\Rightarrow\Sigma,\Omega,\Omega$}
	\doubleLine
	\UIC{$ \Gamma,\Theta
		\Rightarrow \Delta,\Lambda  \mid \psi,\Pi,\Xi\Rightarrow\Sigma,\Omega $}	
\end{prooftree}

Then we reason as follows:
\begin{prooftree}
	\EnableBpAbbreviations
	\AXC{$\Gamma\Rightarrow\Delta\mid \varphi, \Pi\Rightarrow\Sigma, \psi$}
	\AXC{$ \Gamma,\Theta
		\Rightarrow \Delta,\Lambda  \mid \psi,\Pi,\Xi\Rightarrow\Sigma,\Omega $}
	\BIC{$\Gamma,\Gamma,\Theta
		\Rightarrow \Delta,\Delta,\Lambda  \mid \varphi,\Pi,\Pi,\Xi\Rightarrow\Sigma,\Sigma,\Omega$}\doubleLine
	\UIC{$\Gamma,\Theta
		\Rightarrow \Delta,\Lambda  \mid \varphi,\Pi,\Xi\Rightarrow\Sigma,\Omega$}	
\end{prooftree}

Finally, the last step:
\begin{prooftree}
	\EnableBpAbbreviations
	\AXC{$\Theta\Rightarrow\Lambda \mid \Xi\Rightarrow\Omega, \varphi$}	
	\AXC{$\Gamma,\Theta
		\Rightarrow \Delta,\Lambda  \mid \varphi,\Pi,\Xi\Rightarrow\Sigma,\Omega$}
	\BIC{$\Gamma,\Theta,\Theta
		\Rightarrow \Delta,\Lambda,\Lambda  \mid \Pi,\Xi,\Xi\Rightarrow\Sigma,\Omega,\Omega$}\doubleLine
	\UIC{$\Gamma,\Theta
		\Rightarrow \Delta,\Lambda  \mid \Pi,\Xi\Rightarrow\Sigma,\Omega$}
\end{prooftree}

Let us consider the following application of $ (\mid Cut) $ to \L ukasiewicz's implication (with the use of $(\mathord{\mid\rightarrow}\mathord{\Rightarrow})$ and $(\mathord{\mid\Rightarrow}\mathord{\rightarrow})$):
 \begin{prooftree}
		\EnableBpAbbreviations
		\AXC{$ \varphi, \Gamma
			\Rightarrow \Delta \mid\Pi\Rightarrow\Sigma, \psi $}
		\UIC{$ \Gamma\Rightarrow\Delta \mid \Pi\Rightarrow\Sigma, \varphi\rightarrow\psi $}	
		\AXC{$\Theta
			\Rightarrow \Lambda, \varphi \mid \Xi\Rightarrow\Omega$}
		\AXC{$\Theta\Rightarrow\Lambda \mid \psi,
			\Xi \Rightarrow \Omega$}
		\BIC{$\Theta
			\Rightarrow \Lambda \mid \varphi\rightarrow\psi, \Xi\Rightarrow\Omega$}
		\BIC{$\Gamma,\Theta
			\Rightarrow \Delta,\Lambda \mid  \Pi,\Xi\Rightarrow\Sigma,\Omega$}
	\end{prooftree}
	
	We replace this deduction with the subsequent one:
 \begin{prooftree}
 	\EnableBpAbbreviations
 	\AXC{$\Theta
 		\Rightarrow \Lambda, \varphi \mid \Xi\Rightarrow\Omega$}
 	\AXC{$ \varphi, \Gamma
 		\Rightarrow \Delta \mid\Pi\Rightarrow\Sigma, \psi $}
 	\BIC{$\Gamma,\Theta
 		\Rightarrow \Delta,\Lambda \mid  \Pi,\Xi\Rightarrow\Sigma,\Omega,\psi $}	
	\AXC{$\Theta\Rightarrow\Lambda \mid \psi,
		\Xi \Rightarrow \Omega$}
	\BIC{$\Gamma,\Theta,\Theta
		\Rightarrow \Delta,\Lambda,\Lambda \mid  \Pi,\Xi,\Xi\Rightarrow\Sigma,\Omega,\Omega$}\doubleLine
	\UIC{$\Gamma,\Theta
		\Rightarrow \Delta,\Lambda \mid  \Pi,\Xi\Rightarrow\Sigma,\Omega$}
 \end{prooftree}

\textit{Case 2.3.} The cut-formula is McCarthy's conjunction. Consider the following applications of $(\mid \mathord{\Rightarrow}\mathord{\wedge_{mC}})$ and $(\mid \mathord{\wedge}\mathord{\Rightarrow_{mC}})$, respectively:
 
\begin{center}
\EnableBpAbbreviations
\AXC{$ \Gamma
	\Rightarrow \Delta \mid \Pi\Rightarrow\Sigma, \varphi $}
\AXC{$ \varphi, \Gamma\Rightarrow\Delta \mid \Pi\Rightarrow\Sigma, \psi $}
\BIC{$ \Gamma\Rightarrow\Delta \mid \Pi\Rightarrow\Sigma, \varphi\wedge\psi $}
\DisplayProof\end{center}

\begin{center}\EnableBpAbbreviations
\AXC{$ \Theta\Rightarrow \Lambda \mid \varphi, \psi, \Xi\Rightarrow\Omega $}	
\AXC{$ \Theta\Rightarrow\Lambda, \varphi \mid \varphi, \Xi\Rightarrow\Omega $}
\BIC{$ \Theta\Rightarrow\Lambda \mid \varphi\wedge\psi, \Xi\Rightarrow\Omega $}
\DisplayProof\end{center}

Then we apply $ (\mid Cut) $ as follows:
\begin{center}\EnableBpAbbreviations
\AXC{$ \Gamma\Rightarrow\Delta \mid \Pi\Rightarrow\Sigma, \varphi\wedge\psi $}
\AXC{$ \Theta\Rightarrow\Lambda \mid \varphi\wedge\psi, \Xi\Rightarrow\Omega $}
\BIC{$ \Gamma,\Theta\Rightarrow\Delta,\Lambda\mid\Pi,\Xi\Rightarrow\Sigma,\Omega $}
\DisplayProof
\end{center} 
 
We transform the deduction as follows:  
\begin{center}
{\small 	\EnableBpAbbreviations
\AXC{$ \Theta\Rightarrow\Lambda, \varphi \mid \varphi, \Xi\Rightarrow\Omega $}
	\AXC{$ \varphi, \Gamma\Rightarrow\Delta \mid \Pi\Rightarrow\Sigma, \psi $}
\AXC{$ \Theta\Rightarrow \Lambda \mid \varphi, \psi, \Xi\Rightarrow\Omega $}
	\BIC{$  \varphi,\Gamma,\Theta\Rightarrow\Delta,\Lambda \mid \varphi,\Pi,\Xi\Rightarrow\Sigma,\Omega $}
\BIC{$  \Gamma,\Theta,\Theta\Rightarrow\Delta,\Lambda,\Lambda \mid \varphi,\varphi,\Pi,\Xi,\Xi\Rightarrow\Sigma,\Omega,\Omega $}\doubleLine
\UIC{$  \Gamma,\Theta\Rightarrow\Delta,\Lambda \mid \varphi,\Pi,\Xi\Rightarrow\Sigma,\Omega$}	
	\DisplayProof}\end{center} 
 
 We continue this transformation as follows:
 \begin{prooftree}
\EnableBpAbbreviations
	\AXC{$ \Gamma
		\Rightarrow \Delta \mid \Pi\Rightarrow\Sigma, \varphi $}
\AXC{$ \Gamma,\Theta\Rightarrow\Delta,\Lambda \mid \varphi,\Pi,\Xi\Rightarrow\Sigma,\Omega $} 
\BIC{$\Gamma,\Gamma,\Theta\Rightarrow\Delta,\Delta,\Lambda \mid \Pi,\Pi,\Xi\Rightarrow\Sigma,\Sigma,\Omega$}\doubleLine
\UIC{$  \Gamma,\Theta\Rightarrow\Delta,\Lambda \mid \Pi,\Xi\Rightarrow\Sigma,\Omega$}	
 \end{prooftree}
 
 \textit{Case 2.4.} The cut-formula is \L ukasiewicz's disjunction, the case of $ (\mid Cut) $. Consider the following deduction:
 \begin{prooftree}
 \EnableBpAbbreviations
 \AXC{$ \Gamma
 	\Rightarrow \Delta \mid \Pi\Rightarrow\Sigma, \varphi\vee\psi $}
\AXC{$ \Theta\Rightarrow\Lambda \mid \varphi\vee\psi, \Xi\Rightarrow\Omega $}
\BIC{$ \Gamma,\Theta\Rightarrow\Delta,\Lambda \mid \Pi,\Xi\Rightarrow\Sigma,\Omega $} 
 \end{prooftree}
where the premisses are obtained in the following ways, by the rules $(\mathord{\mid\Rightarrow\vee_L})$ and $(\mathord{\mid\vee_L\Rightarrow})$, respectively:
\begin{center}
{\small 	$\dfrac{\Gamma
		\Rightarrow \Delta, \varphi, \psi \mid \Pi\Rightarrow\Sigma \qquad\varphi,  \Gamma\Rightarrow\Delta \mid
		\Pi \Rightarrow \Sigma, \psi \qquad \psi, \Gamma\Rightarrow\Delta \mid \Pi\Rightarrow\Sigma, \varphi } {\Gamma
		\Rightarrow \Delta \mid \Pi\Rightarrow\Sigma, \varphi\vee\psi}$}
\end{center}
\begin{center}
	$\dfrac{\varphi, \Theta
		\Rightarrow \Lambda \mid \psi, \Xi\Rightarrow\Omega \qquad \psi, \Theta\Rightarrow\Lambda\mid \varphi, \Xi\Rightarrow\Omega}{\Theta\Rightarrow\Lambda \mid \varphi\vee\psi, \Xi\Rightarrow\Omega}$
\end{center}
 
 We transform this deduction as follows: 
{\small \begin{prooftree}
	\EnableBpAbbreviations
\AXC{$ \Gamma
	\Rightarrow \Delta, \varphi, \psi \mid \Pi\Rightarrow\Sigma $}
	\AXC{$\psi, \Gamma\Rightarrow\Delta \mid \Pi\Rightarrow\Sigma, \varphi$}
	\AXC{$ \psi, \Theta\Rightarrow\Lambda\mid \varphi, \Xi\Rightarrow\Omega $}
	\BIC{$\psi,\psi, \Gamma,\Theta \Rightarrow\Delta,\Lambda\mid \Pi,\Xi\Rightarrow\Sigma,\Omega $}
\BIC{$ \Gamma,\Gamma,\Theta \Rightarrow\Delta,\Delta,\Lambda,\varphi\mid \Pi,\Pi,\Xi\Rightarrow\Sigma,\Sigma,\Omega $}\doubleLine
\UIC{$ \Gamma,\Theta \Rightarrow\Delta,\Lambda,\varphi\mid \Pi,\Xi\Rightarrow\Sigma,\Omega $}
\end{prooftree}}

In parallel, we build the following deduction:
\begin{prooftree}
\EnableBpAbbreviations	
\AXC{$ \varphi,  \Gamma\Rightarrow\Delta \mid
	\Pi \Rightarrow \Sigma, \psi $}
\AXC{$ \varphi, \Theta
	\Rightarrow \Lambda \mid \psi, \Xi\Rightarrow\Omega $}
\BIC{$ \varphi,\varphi,\Gamma,\Theta \Rightarrow\Delta,\Lambda\mid \Pi,\Xi\Rightarrow\Sigma,\Omega  $}
\UIC{$ \varphi,\Gamma,\Theta \Rightarrow\Delta,\Lambda\mid \Pi,\Xi\Rightarrow\Sigma,\Omega $}
\end{prooftree}
 
 Finally, we combine these two deductions into one:
 \begin{prooftree}
\EnableBpAbbreviations
\AXC{$ \Gamma,\Theta \Rightarrow\Delta,\Lambda,\varphi\mid \Pi,\Xi\Rightarrow\Sigma,\Omega $}
\AXC{$ \varphi,\Gamma,\Theta \Rightarrow\Delta,\Lambda\mid \Pi,\Xi\Rightarrow\Sigma,\Omega  $} 
\BIC{$ \Gamma,\Gamma,\Theta,\Theta \Rightarrow\Delta,\Delta,\Lambda,\Lambda\mid \Pi,\Pi,\Xi,\Xi\Rightarrow\Sigma,\Sigma,\Omega,\Omega $}\doubleLine
\UIC{$ \Gamma,\Theta \Rightarrow\Delta,\Lambda\mid \Pi,\Xi\Rightarrow\Sigma,\Omega  $}	
 \end{prooftree}
 
\textit{Case 2.5.} The cut-formula is weak Kleene conjunction, the case of $ (\mid Cut) $. Consider the following deduction:
\begin{prooftree}
\EnableBpAbbreviations
\AXC{$ \Gamma
	\Rightarrow \Delta \mid \Pi\Rightarrow\Sigma, \varphi\wedge\psi $}
\AXC{$ \Theta
	\Rightarrow \Lambda \mid \varphi\wedge\psi, \Xi\Rightarrow\Omega $}	
\BIC{$ \Gamma,\Theta \Rightarrow\Delta,\Lambda\mid \Pi,\Xi\Rightarrow\Sigma,\Omega $}
\end{prooftree}
where premisses are obtained by $(\mathord{\mid\Rightarrow}\mathord{\wedge_w})$ and $(\mathord{\mid\wedge_w}\mathord{\Rightarrow})$, respectively:
\begin{center}
{\small 	$\dfrac{\Gamma
		\Rightarrow \Delta \mid \Pi\Rightarrow\Sigma, \varphi, \psi \qquad \varphi, \Gamma\Rightarrow\Delta \mid \Pi \Rightarrow \Sigma, \psi \qquad \psi, \Gamma\Rightarrow\Delta \mid \Pi\Rightarrow\Sigma, \varphi} {\Gamma
		\Rightarrow \Delta \mid \Pi\Rightarrow\Sigma, \varphi\wedge\psi}$}
\end{center}
\begin{center}
	{\small $\dfrac{\Theta
			\Rightarrow \Lambda \mid \varphi, \psi, \Xi\Rightarrow\Omega \qquad \Theta\Rightarrow\Lambda, \varphi \mid \varphi,
			\Xi \Rightarrow \Omega \qquad \Theta\Rightarrow\Lambda, \psi \mid \psi, \Xi\Rightarrow\Omega} {\Theta
			\Rightarrow \Lambda \mid \varphi\wedge\psi, \Xi\Rightarrow\Omega}$}
\end{center}

We transform this deduction as follows:
{\small \begin{prooftree}
\EnableBpAbbreviations
\AXC{$ \Theta\Rightarrow\Lambda, \varphi \mid \varphi,
	\Xi \Rightarrow \Omega $}
\AXC{$ \varphi, \Gamma\Rightarrow\Delta \mid \Pi \Rightarrow \Sigma, \psi $}	
\BIC{$ \Gamma,\Theta \Rightarrow\Delta,\Lambda\mid \varphi,\Pi,\Xi\Rightarrow\Sigma,\Omega,\psi $}
\AXC{$ \Theta
	\Rightarrow \Lambda \mid \varphi, \psi, \Xi\Rightarrow\Omega $}
\BIC{$ \Gamma,\Theta,\Theta \Rightarrow\Delta,\Lambda,\Lambda\mid \varphi,\varphi,\Pi,\Xi,\Xi\Rightarrow\Sigma,\Omega,\Omega $}\doubleLine
\UIC{$ \Gamma,\Theta \Rightarrow\Delta,\Lambda\mid \varphi,\Pi,\Xi\Rightarrow\Sigma,\Omega $}
\end{prooftree}}

We continue this deduction as follows:
\begin{prooftree}
\EnableBpAbbreviations
\AXC{$ \Gamma
	\Rightarrow \Delta \mid \Pi\Rightarrow\Sigma, \varphi, \psi $}
\AXC{$ \Gamma,\Theta \Rightarrow\Delta,\Lambda\mid \varphi,\Pi,\Xi\Rightarrow\Sigma,\Omega $}
\BIC{$ \Gamma,\Gamma,\Theta \Rightarrow\Delta,\Delta,\Lambda\mid \Pi,\Pi,\Xi\Rightarrow\Sigma,\Sigma,\Omega,\psi $}	\doubleLine
\UIC{$ \Gamma,\Theta \Rightarrow\Delta,\Lambda\mid \Pi,\Xi\Rightarrow\Sigma,\Omega,\psi $}	
\end{prooftree}

In parallel, we build the following deduction:
{\small \begin{prooftree}
\EnableBpAbbreviations
\AXC{$ \Theta\Rightarrow\Lambda, \psi \mid \psi, \Xi\Rightarrow\Omega $}	
\AXC{$ \psi, \Gamma\Rightarrow\Delta \mid \Pi\Rightarrow\Sigma, \varphi $}
\BIC{$ \Gamma,\Theta \Rightarrow\Delta,\Lambda\mid \psi,\Pi,\Xi\Rightarrow\Sigma,\Omega,\varphi $}
\AXC{$ \Theta
	\Rightarrow \Lambda \mid \varphi, \psi, \Xi\Rightarrow\Omega $}
\BIC{$ \Gamma,\Theta,\Theta \Rightarrow\Delta,\Lambda,\Lambda\mid \psi,\psi,\Pi,\Xi,\Xi\Rightarrow\Sigma,\Omega,\Omega $}\doubleLine
\UIC{$ \Gamma,\Theta \Rightarrow\Delta,\Lambda\mid \psi,\Pi,\Xi\Rightarrow\Sigma,\Omega $}
\end{prooftree}}

Finally, we combine these two deductions:
\begin{prooftree}
\EnableBpAbbreviations
\AXC{$ \Gamma,\Theta \Rightarrow\Delta,\Lambda\mid \Pi,\Xi\Rightarrow\Sigma,\Omega,\psi $}
\AXC{$ \Gamma,\Theta \Rightarrow\Delta,\Lambda\mid \psi,\Pi,\Xi\Rightarrow\Sigma,\Omega $}
\BIC{$ \Gamma,\Gamma,\Theta,\Theta \Rightarrow\Delta,\Delta,\Lambda,\Lambda\mid \Pi,\Pi,\Xi,\Xi\Rightarrow\Sigma,\Sigma,\Omega,\Omega $}\doubleLine
\UIC{$ \Gamma,\Theta \Rightarrow\Delta,\Lambda\mid \Pi,\Xi\Rightarrow\Sigma,\Omega $}	
\end{prooftree}

 \textit{Case 2.6.} The cut-formula is Soboci\'{n}ski's implication $ \rightarrow_S^\prime $, the case of $ (Cut\mid) $. Consider the following deduction:
 \begin{prooftree}
\EnableBpAbbreviations
\AXC{$ \Gamma
	\Rightarrow \Delta, \varphi\rightarrow\psi \mid \Pi\Rightarrow\Sigma $}
\AXC{$ \varphi\rightarrow\psi, \Theta
	\Rightarrow \Lambda \mid \Xi\Rightarrow\Omega $} 
\BIC{$ \Gamma,\Theta \Rightarrow\Delta,\Lambda\mid \Pi,\Xi\Rightarrow\Sigma,\Omega $}	
 \end{prooftree}
where premisses are obtained by $(\mathord{\Rightarrow}\mathord{\rightarrow_S'\mid})$ and $(\mathord{\rightarrow_S'}\mathord{\Rightarrow\mid})$, respectively:
\begin{center}
	$\dfrac{\varphi, \Gamma
		\Rightarrow \Delta, \psi \mid \Pi\Rightarrow\Sigma \;\quad \Gamma\Rightarrow\Delta \mid\varphi,
		\Pi \Rightarrow \Sigma, \psi \;\quad \Gamma\Rightarrow\Delta, \psi \mid \psi, \Pi\Rightarrow\Sigma, \varphi } {\Gamma
		\Rightarrow \Delta, \varphi\rightarrow\psi \mid \Pi\Rightarrow\Sigma}$
\end{center}
\begin{center}
{\small 	$\dfrac{\psi, \Theta
		\Rightarrow \Lambda \mid \Xi\Rightarrow\Omega \qquad \Theta\Rightarrow\Lambda \mid \Xi \Rightarrow \Omega, \varphi, \psi \qquad \Theta\Rightarrow\Lambda, \varphi\mid \varphi, \psi, \Xi\Rightarrow\Omega} {\varphi\rightarrow\psi, \Theta
		\Rightarrow \Lambda \mid \Xi\Rightarrow\Omega}$}
\end{center}

We build a new deduction, using this material:
{\small \begin{prooftree}
	\EnableBpAbbreviations
\AXC{$ \Theta\Rightarrow\Lambda \mid \Xi \Rightarrow \Omega, \varphi, \psi $}
\AXC{$\Gamma\Rightarrow\Delta, \psi \mid \psi, \Pi\Rightarrow\Sigma, \varphi$}
\AXC{$ \psi, \Theta
	\Rightarrow \Lambda \mid \Xi\Rightarrow\Omega $}
\BIC{$ \Gamma,\Theta \Rightarrow\Delta,\Lambda\mid \psi,\Pi,\Xi\Rightarrow\Sigma,\Omega,\varphi $}
\BIC{$ \Gamma,\Theta,\Theta \Rightarrow\Delta,\Lambda,\Lambda\mid \Pi,\Xi,\Xi\Rightarrow\Sigma,\Omega,\Omega,\varphi,\varphi $}\doubleLine	
\UIC{$ \Gamma,\Theta \Rightarrow\Delta,\Lambda\mid \Pi,\Xi\Rightarrow\Sigma,\Omega,\varphi $}	
\end{prooftree}}
 
In parallel, we build yet another deduction:
{\small \begin{prooftree}
	\EnableBpAbbreviations
\AXC{$ \Theta\Rightarrow\Lambda, \varphi\mid \varphi, \psi, \Xi\Rightarrow\Omega $}
\AXC{$ \varphi, \Gamma
	\Rightarrow \Delta, \psi \mid \Pi\Rightarrow\Sigma $}
\BIC{$ \Gamma,\Theta \Rightarrow\Delta,\Lambda,\psi\mid \varphi, \psi,\Pi,\Xi\Rightarrow\Sigma,\Omega $}
\AXC{$ \psi, \Theta
	\Rightarrow \Lambda \mid \Xi\Rightarrow\Omega $}
\BIC{$ \Gamma,\Theta,\Theta \Rightarrow\Delta,\Lambda,\Lambda\mid \varphi, \psi,\Pi,\Xi,\Xi\Rightarrow\Sigma,\Omega,\Omega $}\doubleLine
\UIC{$ \Gamma,\Theta \Rightarrow\Delta,\Lambda\mid \varphi, \psi,\Pi,\Xi\Rightarrow\Sigma,\Omega $}	
\end{prooftree}} 

We continue this deduction as follows:
\begin{prooftree}
\EnableBpAbbreviations	
\AXC{$ \Gamma\Rightarrow\Delta \mid\varphi,
	\Pi \Rightarrow \Sigma, \psi $}
\AXC{$ \Gamma,\Theta \Rightarrow\Delta,\Lambda\mid \varphi, \psi,\Pi,\Xi\Rightarrow\Sigma,\Omega $}
\BIC{$\Gamma,\Gamma,\Theta \Rightarrow\Delta,\Delta,\Lambda\mid \varphi,\varphi, \Pi,\Pi,\Xi\Rightarrow\Sigma,\Sigma,\Omega  $}\doubleLine
\UIC{$\Gamma,\Theta \Rightarrow\Delta,\Lambda\mid \varphi, \Pi,\Xi\Rightarrow\Sigma,\Omega  $}
\end{prooftree}

Finally, we combine these two deductions:
\begin{prooftree}
\EnableBpAbbreviations
\AXC{$ \Gamma,\Theta \Rightarrow\Delta,\Lambda\mid \Pi,\Xi\Rightarrow\Sigma,\Omega,\varphi $}
\AXC{$\Gamma,\Theta \Rightarrow\Delta,\Lambda\mid \varphi, \Pi,\Xi\Rightarrow\Sigma,\Omega  $}	
\BIC{$\Gamma,\Gamma,\Theta,\Theta \Rightarrow\Delta,\Delta,\Lambda,\Lambda\mid \Pi,\Pi,\Xi,\Xi\Rightarrow\Sigma,\Sigma,\Omega,\Omega  $}\doubleLine
\UIC{$\Gamma,\Theta \Rightarrow\Delta,\Lambda\mid \Pi,\Xi\Rightarrow\Sigma,\Omega  $}
\end{prooftree} 
 
 \textbf{Case 3.} The cut-formula is not principal in at least one premiss. Let it be the
 right premiss.
 
\textit{Case 3.1.} The right premiss is deduced by $ (\neg\Rightarrow \mid) $. The rule $ (Cut\mid) $ is applied:

\begin{prooftree}
	\EnableBpAbbreviations
	\AXC{$\Gamma\Rightarrow\Delta,\varphi\mid\Pi\Rightarrow\Sigma$}
	\AXC{$\varphi,\Theta\Rightarrow\Lambda\mid\Xi\Rightarrow\Omega,\psi$}
	\RL{$ (\neg\Rightarrow\mid) $}
	\UIC{$\varphi,\neg\psi,\Theta\Rightarrow\Lambda\mid\Xi\Rightarrow\Omega$}
	\RL{$ (Cut\mid) $}
	\BIC{$\neg\psi,\Gamma,\Theta\Rightarrow\Delta,\Lambda\mid\Pi,\Xi\Rightarrow\Sigma,\Omega$}	
\end{prooftree}

We transform it as follows:
\begin{prooftree}
	\EnableBpAbbreviations
	\AXC{$\Gamma\Rightarrow\Delta,\varphi\mid\Pi\Rightarrow\Sigma$}
	\AXC{$\varphi,\Theta\Rightarrow\Lambda\mid\Xi\Rightarrow\Omega,\psi$}
	\RL{$ (Cut\mid) $}	
	\BIC{$\Gamma,\Theta\Rightarrow\Delta,\Lambda\mid\Pi,\Xi\Rightarrow\Sigma,\Omega,\psi$}
\RL{$ (\neg\Rightarrow\mid) $}
	\UIC{$\neg\psi,\Gamma,\Theta\Rightarrow\Delta,\Lambda\mid\Pi,\Xi\Rightarrow\Sigma,\Omega$}	
\end{prooftree}

The right premiss is deduced by $ (\mid\neg\Rightarrow) $. The rule $ (\mid Cut) $ is applied:
\begin{prooftree}
	\EnableBpAbbreviations
	\AXC{$\Gamma\Rightarrow\Delta\mid\Pi\Rightarrow\Sigma,\varphi$}
	\AXC{$\Theta\Rightarrow\Lambda,\psi\mid\varphi,\Xi\Rightarrow\Omega$}
	\RL{$ (\mid\neg\Rightarrow) $}
	\UIC{$\Theta\Rightarrow\Lambda\mid\varphi,\neg\psi,\Xi\Rightarrow\Omega$}
	\RL{$ (\mid Cut) $}
	\BIC{$\Gamma,\Theta\Rightarrow\Delta,\Lambda\mid\neg\psi,\Pi,\Xi\Rightarrow\Sigma,\Omega$}	
\end{prooftree}

We transform this derivation as follows:
\begin{prooftree}
	\EnableBpAbbreviations
	\AXC{$\Gamma\Rightarrow\Delta\mid\Pi\Rightarrow\Sigma,\varphi$}
	\AXC{$\Theta\Rightarrow\Lambda,\psi\mid\varphi,\Xi\Rightarrow\Omega$}
	\RL{$ (\mid Cut) $}
	\BIC{$\Gamma,\Theta\Rightarrow\Delta,\Lambda,\psi\mid\Pi,\Xi\Rightarrow\Sigma,\Omega$}		
	\RL{$ (\mid\neg\Rightarrow) $}
	\UIC{$\Gamma,\Theta\Rightarrow\Delta,\Lambda\mid\neg\psi,\Pi,\Xi\Rightarrow\Sigma,\Omega$}
\end{prooftree}

\textit{Case 3.2.} The right premiss is deduced by $(\mathord{\rightarrow_L}\mathord{\Rightarrow}\mid)$:
\begin{center}
{\small 	$ \dfrac{\varphi, \Theta
		\Rightarrow \Lambda, \chi \mid \psi, \Xi\Rightarrow\Omega \qquad \psi,\varphi, \Theta\Rightarrow\Lambda \mid
		\Xi \Rightarrow \Omega \qquad \varphi,\Theta\Rightarrow\Lambda \mid \Xi\Rightarrow\Omega, \chi}{\chi\rightarrow\psi,\varphi, \Theta
		\Rightarrow \Lambda \mid \Xi\Rightarrow\Omega} $}
\end{center}

Then the rule $ (Cut\mid) $ is applied:
\begin{center}
	$ \dfrac{\Gamma\Rightarrow\Delta,\varphi\mid\Pi\Rightarrow\Sigma \qquad \varphi,\chi\rightarrow\psi, \Theta
		\Rightarrow \Lambda \mid \Xi\Rightarrow\Omega}{\chi\rightarrow\psi,\Gamma, \Theta
		\Rightarrow \Delta,\Lambda \mid \Pi,\Xi\Rightarrow\Sigma,\Omega} $
\end{center}

We transform this deduction as follows, we start with three applications of $ (Cut\mid) $:
\begin{prooftree}
	\EnableBpAbbreviations
\AXC{$\Gamma\Rightarrow\Delta,\varphi\mid\Pi\Rightarrow\Sigma$}
\AXC{$\varphi, \Theta
	\Rightarrow \Lambda, \chi \mid \psi, \Xi\Rightarrow\Omega$}	
\BIC{$\Gamma, \Theta
	\Rightarrow \Delta,\Lambda,\chi \mid \psi,\Pi,\Xi\Rightarrow\Sigma,\Omega$}
\end{prooftree}

\begin{center}
\EnableBpAbbreviations
\AXC{$\Gamma\Rightarrow\Delta,\varphi\mid\Pi\Rightarrow\Sigma$}
\AXC{$\psi,\varphi,\Theta\Rightarrow\Lambda \mid
	\Xi \Rightarrow \Omega$}
\BIC{$\psi,\Gamma, \Theta
	\Rightarrow \Delta,\Lambda \mid \Pi,\Xi\Rightarrow\Sigma,\Omega$}
\DisplayProof\end{center}\begin{center}
\EnableBpAbbreviations
\AXC{$\Gamma\Rightarrow\Delta,\varphi\mid\Pi\Rightarrow\Sigma$}
\AXC{$\varphi,\Theta\Rightarrow\Lambda \mid \Xi\Rightarrow\Omega, \chi$}
\BIC{$\Gamma, \Theta
	\Rightarrow \Delta,\Lambda \mid \Pi,\Xi\Rightarrow\Sigma,\Omega, \chi$}
\DisplayProof
\end{center}

Finally, we apply $(\mathord{\rightarrow_L}\mathord{\Rightarrow}\mid)$ to the results of these applications of $ (Cut\mid) $ and get the required result.

The right premiss is deduced by $(\mathord{\rightarrow_L}\mathord{\Rightarrow}\mid)$:
 \begin{prooftree}
 	\EnableBpAbbreviations
 	\AXC{$ \Gamma\Rightarrow\Delta \mid \Pi\Rightarrow\Sigma, \varphi $}	
 	\AXC{$\Theta
 		\Rightarrow \Lambda, \chi \mid \varphi,\Xi\Rightarrow\Omega$}
 	\AXC{$\Theta\Rightarrow\Lambda \mid \psi,
 		\varphi,\Xi \Rightarrow \Omega$}
 	\BIC{$\Theta
 		\Rightarrow \Lambda \mid\varphi, \chi\rightarrow\psi, \Xi\Rightarrow\Omega$}
 	\BIC{$\Gamma,\Theta
 		\Rightarrow \Delta,\Lambda \mid  \chi\rightarrow\psi,\Pi,\Xi\Rightarrow\Sigma,\Omega$}
 \end{prooftree}
 
 We transform this deduction as follows. We start with two applications of $ (Cut\mid) $:
 \begin{center}
 	\EnableBpAbbreviations
 	\AXC{$ \Gamma\Rightarrow\Delta \mid \Pi\Rightarrow\Sigma, \varphi $}	
 	\AXC{$\Theta
 		\Rightarrow \Lambda, \chi \mid \varphi,\Xi\Rightarrow\Omega$}
\BIC{$\Gamma,\Theta
	\Rightarrow \Delta,\Lambda,\chi \mid  \Pi,\Xi\Rightarrow\Sigma,\Omega$} 
\DisplayProof \end{center} 
\begin{center}	\EnableBpAbbreviations
 	\AXC{$ \Gamma\Rightarrow\Delta \mid \Pi\Rightarrow\Sigma, \varphi $}	
 	\AXC{$\Theta\Rightarrow\Lambda \mid \psi,
 		\varphi,\Xi \Rightarrow \Omega$}
 	\BIC{$\Gamma,\Theta
 		\Rightarrow \Delta,\Lambda \mid  \psi,\Pi,\Xi\Rightarrow\Sigma,\Omega$}
\DisplayProof  		
 \end{center}
 
 Then we apply $(\mathord{\rightarrow_L}\mathord{\Rightarrow}\mid)$: 
  \begin{center}
  	\EnableBpAbbreviations
  	\AXC{$\Gamma,\Theta
  		\Rightarrow \Delta,\Lambda,\chi \mid  \Pi,\Xi\Rightarrow\Sigma,\Omega$} 
  	\AXC{$\Gamma,\Theta
  		\Rightarrow \Delta,\Lambda \mid  \psi,\Pi,\Xi\Rightarrow\Sigma,\Omega$}
  	\BIC{$\Gamma,\Theta
  		\Rightarrow \Delta,\Lambda \mid  \chi\rightarrow\psi,\Pi,\Xi\Rightarrow\Sigma,\Omega$}
  	\DisplayProof 	
  \end{center} 

\textit{Case 3.3.} The right premiss is deduced by  $(\mid \mathord{\wedge}\mathord{\Rightarrow_{mC}})$:

\begin{center}
{\small 	\EnableBpAbbreviations
	\AXC{$ \Gamma\Rightarrow\Delta \mid \Pi\Rightarrow\Sigma, \varphi$}
	\AXC{$ \Theta\Rightarrow \Lambda \mid \chi, \psi, \varphi,\Xi\Rightarrow\Omega $}	
	\AXC{$ \Theta\Rightarrow\Lambda, \chi \mid \chi, \varphi,\Xi\Rightarrow\Omega $}
	\BIC{$ \Theta\Rightarrow\Lambda \mid \chi\wedge\psi,\varphi, \Xi\Rightarrow\Omega $}
\BIC{$\Gamma,\Theta\Rightarrow\Delta,\Lambda \mid\chi\wedge\psi,  \Pi,\Xi\Rightarrow\Sigma,\Omega$}	
	\DisplayProof}\end{center}

Then we apply $ (\mid Cut) $ twice as follows:
\begin{center}\EnableBpAbbreviations
	\AXC{$ \Gamma\Rightarrow\Delta \mid \Pi\Rightarrow\Sigma, \varphi$}
	\AXC{$ \Theta\Rightarrow \Lambda \mid \chi, \psi, \varphi,\Xi\Rightarrow\Omega $}
\BIC{$\Gamma,\Theta\Rightarrow\Delta,\Lambda \mid\chi,\psi,  \Pi,\Xi\Rightarrow\Sigma,\Omega$}
	\DisplayProof
\end{center} 
\begin{center}\EnableBpAbbreviations
	\AXC{$ \Gamma\Rightarrow\Delta \mid \Pi\Rightarrow\Sigma, \varphi$}
	\AXC{$ \Theta\Rightarrow\Lambda, \chi \mid \chi, \varphi,\Xi\Rightarrow\Omega $}
\BIC{$\Gamma,\Theta\Rightarrow\Delta,\Lambda,\chi \mid\chi,\Pi,\Xi\Rightarrow\Sigma,\Omega$}
	\DisplayProof
\end{center} 

Finally, we apply $(\mid \mathord{\wedge}\mathord{\Rightarrow_{mC}})$ to the results of these applications of $ (\mid Cut) $ and get the required result. 
  
  The other cases are considered similarly.
\end{proof}

\section{Interpolation}\label{Interpolation}

We present a constructive proof of the interpolation theorem for some logics 
based on the strategy
proposed by Muskens and Wintein \cite{win:int17}. It was originally applied in the tableau setting for Belnap-Dunn four-valued logic, as well as for ${\bf K_3}$ and {\bf LP}. We applied it also to BSC formalisation of several four-valued quasi-relevant logics \cite{Indrzejczak3}. Here we demonstrate that it can be used
for showing that interpolation holds for some three-valued paracomplete and paraconsistent logics.
Let ${\bf L}\in\{{\bf I^1, I^2, P^1, P^2}\}$.

\begin{theorem}
	For any contingent formulae $\varphi, \psi$, if $\varphi\models_L\psi$, then we can construct an interpolant for ${\bf I^1, I^2}$ on the basis of
	proof-search trees for $\varphi\Rightarrow \mid \Rightarrow$ and $\Rightarrow\psi\mid\Rightarrow$ and an interpolant for ${\bf P^1, P^2}$ on the basis of
	proof-search trees for $\Rightarrow\mid\varphi\Rightarrow$ and $\Rightarrow\mid\Rightarrow\psi$ in suitable BSC-L.
\end{theorem}

 \begin{proof}
 	We will demonstrate the case of BSC-I$^1$; the case of BSC-I$^2$ is identical and the cases of BSC-P$^1$ and BSC-P$^2$ are dual, so we only comment on them in the key points. Assume that $\varphi\models_{I^1}\psi$; hence by completeness we have a cut-free proof of $\varphi\Rightarrow\psi\mid\Rightarrow$ in BSC-I$^1$. Now produce complete proof-search trees for
$\varphi\Rightarrow\mid\Rightarrow$ and $\Rightarrow\psi\mid\Rightarrow$. Since $\varphi, \psi$ are contingent, they have some nonaxiomatic leaves. Let $\Gamma_1\Rightarrow\Delta_1\mid\Pi_1\Rightarrow\Sigma_1, \ldots, \Gamma_k\Rightarrow\Delta_k\mid\Pi_k\Rightarrow\Sigma_k$ be the list
of nonaxiomatic atomic leaves of the proof-search tree for $\varphi\Rightarrow\mid\Rightarrow$ and $\Theta_1\Rightarrow\Lambda_1\mid\Xi_1\Rightarrow\Omega_1, \ldots, \Theta_n\Rightarrow\Lambda_n\mid\Xi_n\Rightarrow\Omega_n$ such a list
taken from the proof-search tree for $\Rightarrow\psi\mid\Rightarrow$. It holds: 

\begin{claim}[1]
	For any $i\leq k$ and $j\leq n$, $\Gamma_i, \Theta_j\Rightarrow\Delta_i, \Lambda_j\mid\Pi_i, \Xi_j\Rightarrow\Sigma_i, \Omega_j$ is an axiomatic clause.
\end{claim}

To see this take a tree for $\varphi\Rightarrow\mid\Rightarrow$ and add $\psi$ to succedents of all 1-sequents in the tree. Due to context independence of all rules
it is a correct proof-search tree. Now to each leaf $\Gamma_i\Rightarrow\Delta_i, \psi\mid\Pi_i\Rightarrow\Sigma_i$ append a tree of $\Rightarrow\psi\mid\Rightarrow$ but with $\Gamma_i$ added to each antecedent
and $\Delta_i$ added to each succedent of 1-sequents, and similarly with $\Pi_i$ and $\Sigma_i$ in all 2-sequents. In the resulting proof-search tree we have leaves of the form $\Gamma_i, \Theta_j\Rightarrow\Delta_i, \Lambda_j\mid\Pi_i, \Xi_j\Rightarrow\Sigma_i, \Omega_j$
for all $i\leq k$ and $j\leq n$. If at least one of them is not axiomatic, then $\nvdash \varphi\Rightarrow\psi\mid\Rightarrow$. \qed

\1

Next for every $\Gamma_i\Rightarrow\Delta_i\mid\Pi_i\Rightarrow\Sigma_i, i\leq k$, define the following sets:

\1

$\Gamma_i' = \Gamma_i\cap (\bigcup\Lambda_j\cup\bigcup\Omega_j)$ for $j\leq n$

$\Delta_i' = \Delta_i\cap\bigcup\Theta_j$ for $j\leq n$

$\Pi_i' = \Pi_i\cap\bigcup\Omega_j$ for $j\leq n$

$\Sigma_i' = \Sigma_i\cap (\bigcup\Theta_j\cup\bigcup\Xi_j)$ for $j\leq n$

\1

Since every
$\Gamma_i, \Theta_j\Rightarrow\Delta_i, \Lambda_j\mid\Pi_i, \Xi_j\Rightarrow\Sigma_i, \Omega_j$ is axiomatic we are guaranteed that $\Gamma_i'\cup\Delta_i'\cup\Pi_i'\cup\Sigma_i'\neq \varnothing$.
Note also that $AT(\Gamma_i'\cup\Delta_i'\cup\Pi_i'\cup\Sigma_i') \subseteq AT(\varphi)\cap AT(\psi)$.

Now define an interpolant $Int(\varphi, \psi)$ for considered logics. For ${\bf I^1, I^2}$  it has the same form:

\1

\noindent $\bigwedge\Gamma_1' \wedge \bigwedge\neg\Sigma_1'\wedge \neg(\bigvee\neg\Pi_1'\vee\bigvee\Delta_1') \ \vee \ldots \vee \ \bigwedge\Gamma_k' \wedge \bigwedge\neg\Sigma_k'\wedge \neg(\bigvee\neg\Pi_k'\vee\bigvee\Delta_k')$, 

\1

\noindent where $\neg\Pi$ means the set of negations of all elements in $\Pi$.

\1

For ${\bf P^1, P^2}$ $Int(\varphi, \psi)$ is defined as:

\1

$\bigwedge\Pi_1' \wedge \bigwedge\neg\Delta_1'\wedge \neg(\bigvee\neg\Gamma_1'\vee\bigvee\Sigma_1') \ \vee \ldots \vee \ \bigwedge\Pi_k' \wedge \bigwedge\neg\Delta_k'\wedge \neg(\bigvee\neg\Gamma_k'\vee\bigvee\Sigma_k')$

\1

We can show that:

\begin{claim}[2]
	$Int(\varphi, \psi)$ is an interpolant for $\varphi\models_L\psi$. \qed
\end{claim}
\begin{proof}
As an example, we present the proof for BSC-I$^1$.  For the sake of proof let us recall that BSC-I$^1$ consists of the rules characterising $\wedge_C, \vee_C, \rightarrow_C$ and $\neg_H$. However, most of the rules necessary for conducting the proof are identical with respective rules from BSC-K$_3$, so the label $C$ in their names will be omitted in these cases for easier recognition where the specific rules (concretely  $(\mid\vee_C\Rightarrow)$ and $(\mid\Rightarrow\vee_C)$) are required.

Since for every $\bigwedge\Gamma_i'\wedge \bigwedge\neg\Sigma_i'\wedge \neg(\bigvee\neg\Pi_i'\vee\bigvee\Delta_i')$ all (negated) atoms are by definition taken from  $AT(\varphi)\cap AT(\psi)$ we must
only prove that BSC-I$^1$ $\vdash\varphi\Rightarrow Int(\varphi, \psi)\mid\Rightarrow$, and BSC-I$^1$ $\vdash Int(\varphi, \psi) \Rightarrow\psi\mid\Rightarrow$ (the same for BSC-I$^2$), and BSC-P$^1$ $\vdash\Rightarrow\mid\varphi\Rightarrow Int(\varphi, \psi)$ and BSC-P$^1$ $\vdash \Rightarrow \mid Int(\varphi, \psi) \Rightarrow\psi$ (and the same for BSC-P$^2$). 

Again take a complete proof-search tree for $\varphi\Rightarrow\mid\Rightarrow$ and add $Int(\varphi, \psi)$ to every succedent of 1-sequent. For every $\Gamma_i\Rightarrow\Delta_i, Int(\varphi, \psi)\mid\Pi_i\Rightarrow\Sigma_i$
apply $(\Rightarrow\vee\mid)$ to get $$\Gamma_i\Rightarrow\Delta_i, \bigwedge\Gamma_i'\wedge \bigwedge\neg\Sigma_i'\wedge \neg(\bigvee\neg\Pi_i'\vee\bigvee\Delta_i'), Int(\varphi, \psi)^{-i}\mid\Pi_i\Rightarrow\Sigma_i,$$ where $Int(\varphi, \psi)^{-i}$
is the rest of the disjunction (if any). Applying $(\Rightarrow\wedge\mid)$  we obtain three bisequents:

\1

(a) $\Gamma_i\Rightarrow\Delta_i, \bigwedge\Gamma_i', Int(\varphi, \psi)^{-i}\mid\Pi_i\Rightarrow\Sigma_i$

(b) $\Gamma_i\Rightarrow\Delta_i, \bigwedge\neg\Sigma_i', Int(\varphi, \psi)^{-i}\mid\Pi_i\Rightarrow\Sigma_i$

(c) $\Gamma_i\Rightarrow\Delta_i, \neg(\bigvee\neg\Pi_i'\vee\bigvee\Delta_i'), Int(\varphi, \psi)^{-i}\mid\Pi_i\Rightarrow\Sigma_i$.

\1

Systematically applying $(\Rightarrow\wedge\mid)$  to (a) we obtain 
$$\Gamma_i\Rightarrow\Delta_i, p, Int(\varphi, \psi)^{-i}\mid\Pi_i\Rightarrow\Sigma_i$$ for each $p\in \Gamma_i'$ and since $\Gamma_i'\subseteq\Gamma_i$ they are all axiomatic. Similarly with (b) but now we first obtain $\Gamma_i\Rightarrow\Delta_i, \neg p, Int(\varphi, \psi)^{-i}\mid\Pi_i\Rightarrow\Sigma_i$ for each $p\in \Sigma_i'$. After the application of $(\Rightarrow\neg\mid)$ we obtain $\Gamma_i\Rightarrow\Delta_i, Int(\varphi, \psi)^{-i}\mid p, \Pi_i\Rightarrow\Sigma_i$ which is axiomatic since $\Sigma_i'\subseteq\Sigma_i$.

For (c) we first apply $(\Rightarrow\neg\mid)$ and obtain $$\Gamma_i\Rightarrow\Delta_i, Int(\varphi, \psi)^{-i}\mid \bigvee\neg\Pi_i'\vee\bigvee\Delta_i', \Pi_i\Rightarrow\Sigma_i$$. By $(\mid\vee_C\Rightarrow)$ we obtain:

$$\bigvee\neg\Pi_i', \Gamma_i\Rightarrow\Delta_i, Int(\varphi, \psi)^{-i}\mid  \Pi_i\Rightarrow\Sigma_i$$ and 

$$\bigvee\Delta_i',\Gamma_i\Rightarrow\Delta_i, Int(\varphi, \psi)^{-i}\mid  \Pi_i\Rightarrow\Sigma_i$$

Systematic application of $(\vee\Rightarrow\mid)$ to the latter produces axiomatic bisequents $$p,\Gamma_i\Rightarrow\Delta_i, Int(\varphi, \psi)^{-i}\mid  \Pi_i\Rightarrow\Sigma_i$$ for each $p\in \Delta_i'$. Systematic application of $(\vee\Rightarrow\mid)$ to the former produces $\neg p,\Gamma_i\Rightarrow\Delta_i, Int(\varphi, \psi)^{-i}\mid  \Pi_i\Rightarrow\Sigma_i$ for each $p\in \Pi_i'$. After application of $(\neg\Rightarrow\mid)$ they also yield axiomatic sequents.
Hence we have a proof of $\varphi\Rightarrow Int(\varphi, \psi)\mid\Rightarrow$.

We have to do the same with a complete proof-search tree for $\Rightarrow\psi\mid\Rightarrow$ but now adding $Int(\varphi, \psi)$ to every antecedent of all 1-sequents in the tree. For every leaf 
$Int(\varphi, \psi), \Theta_j\Rightarrow\Lambda_j\mid\Xi_j\Rightarrow\Omega_j$  we apply $(\vee\Rightarrow|)$ to each disjunct of $Int(\varphi, \psi)$ until we get leaves: $\bigwedge\Gamma_1'\wedge \bigwedge\neg\Sigma_1'\wedge \neg(\bigvee\neg\Pi_1'\vee\bigvee\Delta_1'), \Theta_j\Rightarrow\Lambda_j\mid\Xi_j\Rightarrow\Omega_j$ \ldots 
$\bigwedge\Gamma_k'\wedge \bigwedge\neg\Sigma_k'\wedge \neg(\bigvee\neg\Pi_k'\vee\bigvee\Delta_k'), \Theta_j\Rightarrow\Lambda_j\mid\Xi_j\Rightarrow\Omega_j$.
To each such leaf we apply $(\wedge\Rightarrow\mid)$ obtaining bisequents of the form $\Gamma_i', \neg\Sigma_i', \neg(\bigvee\neg\Pi_i'\vee\bigvee\Delta_i'), \Theta_j\Rightarrow\Lambda_j\mid\Xi_j\Rightarrow\Omega_j$ for $i\leq k, j\leq n$.
In each case the application of $(\neg\Rightarrow\mid)$ yields
$\Gamma_i', \Theta_j\Rightarrow\Lambda_j\mid\Xi_j\Rightarrow\Omega_j, \Sigma_i', \bigvee\neg\Pi_i'\vee\bigvee\Delta_i'$. 
The application of $(\mid\Rightarrow\vee_C)$ to  $\bigvee\neg\Pi_i'\vee\bigvee\Delta_i' $ yields
$\Gamma_i', \Theta_j\Rightarrow\Lambda_j,  \bigvee\neg\Pi_i', \bigvee\Delta_i'\mid\Xi_j\Rightarrow\Omega_j, \Sigma_i'$. Systematic application of $(\Rightarrow\vee\mid)$ and $(\Rightarrow\neg\mid)$ gives leaves of the form
$\Gamma_i', \Theta_j\Rightarrow\Lambda_j,  \Delta_i'\mid\Pi_i', \Xi_j\Rightarrow\Omega_j, \Sigma_i'$. Since for every $i\leq k, j\leq n$, $\Gamma_i, \Theta_j\Rightarrow\Delta_i, \Lambda_j\mid\Pi_i, \Xi_j\Rightarrow\Sigma_i, \Omega_j$ is axiomatic these primed versions are 
axiomatic too. Assume the contrary, then it must be e.g. some $p\notin \Gamma_i'$ such that either $p\in\Gamma_i\cap\Lambda_j $ or $p\in \Gamma_i\cap \Omega_j$ (or for other pairs generating axioms). But it is impossible since by definition $\Gamma_i'$ must contain such $p$ (and the same for other cases of primed sets).

The proof for BSC-I$^2$ is identical since the only difference between these two logics is that ${\bf I^1}$ has Heyting's negation whereas in ${\bf I^2}$ it is Kleene's negation. But the two BSC rules for negation which are used in the proof are common to both negations.

The proof for ${\bf P^1, P^2}$ is dual to the above and uses the slightly different definition of $Int(\varphi, \psi)$ specified above. Again the two logics differ only with respect to negations, but the rules used in the proof are common to Bochvar's and Kleene's ones.	
\end{proof}

Eventually note that this proof may be applied also to other logics but in some cases it is convenient to extend their languages by addition of some other negations. This strategy of proving interpolation theorems for variety of quasi-relevant four-valued logics was examined in \cite{Indrzejczak3}. It can be successfully applied for three-valued logics as well; for example, interpolants for some logics can be defined as disjunctions of the following formulae:

\1

For ${\bf K_3}$ -- $\bigwedge\Gamma_i' \wedge \bigwedge\neg\Sigma_i'\wedge \bigwedge\neg_B\Delta_i'\wedge\bigwedge\neg_B\neg\Pi_i'$

For {\bf LP} -- $\bigwedge\Pi_i'\wedge \bigwedge\neg\Delta_i'\wedge \bigwedge\neg_H\Sigma_i'\wedge \bigwedge\neg_H\neg\Gamma_i'$

For ${\bf G_3}$ -- $\bigwedge\Gamma_i' \wedge \neg_H(\bigwedge\Pi_i'\rightarrow \bigvee \Sigma_i')\wedge\bigwedge\neg_B\Delta_i'$

For ${\bf G_3'}$ -- $\bigwedge\Pi_i'\wedge \bigwedge\neg_B\Delta_i'\wedge \bigwedge\neg_H\Sigma_i'\wedge \bigwedge\neg_B\neg_B\Gamma_i'$
 \end{proof}

\section{Conclusion}\label{Conclusion}

Bisequent calculi can be seen as one of the possible syntactical realizations of so called Suszko's thesis \cite{Suszko} in the treatment of many-valued logics. According to Suszko, every logic is two-valued in the sense that all values are divided into designated and non-designated, and this is reflected in the definition of consequence relation. In the case of bisequent calculi it is additionally made evident that two possible choices of designated values can be made. However, on a deeper level, a BSC is similar to other proposed formalizations mentioned in the Introduction. On one hand, bisequents resemble several labelled approaches where labels denote sets of values; a difference is that instead of labels a position of a formula in a bisequent is crucial, hence the method is strictly syntactical. On the other hand, 
there is a similarity with Avron's \cite{Avron} and Avron, Ben-Naim, and Konikowska's \cite{Avron07} sequent calculi with special rules defined for negated formulae; a difference is that BSC satisfies ordinary subformula property and purity conditions to the effect that in the schemata of rules only one (occurrence of a) connective is involved.  

The price is that instead of standard sequents we use a pair of them. There is one more general difference however. In the case of labelled calculi or Avron's SC we have the same input for 1- and 2-logics, whereas in BSC a different input for both classes of logics is defined; a 1- or a 2-sequent in a bisequent. A consequence of our choice is that for every pair of 1- and 2-logics with the same connectives (like e.g. {\bf K$_3$} and {\bf LP}) the rules and axioms are identical. In contrast, in other mentioned approaches for such pairs of related logics, the respective calculi must differ either with respect to some axioms (closure conditions in tableaux) or to rules. It seems that our solution where systems differ only with respect to the input is more natural and uniform.

Notice that the application of BSC may be extended in several ways. First, it may be applied also to several four-valued logics like {\bf B$_4$} proposed by Belnap \cite{Belnap77a,Belnap77b} and Dunn \cite{Dunn76}. Here in addition to 1, 0 we use $\bot $ for a gap (no value) and $\top $ for a glut (both). Designated values are 1 and $\top $.
The interpretation of bisequents in this setting is: $\Gamma\Rightarrow\Delta \mid \Pi\Rightarrow\Sigma$ is falsified by $h$ iff all elements of $\Gamma$ are either true or $\top$, all elements of $\Delta$ are either false or $\bot$, all elements of $\Pi$ are either true or $\bot$ and all elements of $\Sigma$ are either false or $\top$.  BSC for {\bf B$_4$} is like that for {\bf K$_3$} with two differences. Axioms are of the form $\Gamma\Rightarrow\Delta \mid \Pi\Rightarrow\Sigma$ with nonempty  $\Gamma\cap\Delta$ or $\Pi\cap\Sigma$ (the case of nonempty $\Gamma\cap\Sigma$ is not an axiom). 
We mentioned that the examination of BSC for a family of quasi-relevant four-valued logics built on the basis of {\bf B$_4$} has been already provided \cite{Indrzejczak3}. However, in contrast to the case of three-valued logics, in the field of four-valued ones we are not able to provide a fully uniform approach which covers logics having only one or three designated values.

Finally notice that the application of BSC may be extended to first-order languages.  In particular, we finish this paper with signalling a problem which is under investigation: the application of first-order BSC to formalisation of neutral free logics, and in particular to specific theories of definite descriptions based on Fregean ideas (see e.g. Lehmann \cite{Lehmann}, Stenlund \cite{Stenlund}). 
Sequent and tableau calculi for theories of definite descriptions built on positive and negative free logics have already been provided in \cite{Indrzejczak2020b,Indrzejczak2021c,IndZaw1}, but neutral free logic as the basis has not been investigated so far. Some proof systems for neutral free logics in the basic language with simple terms have only been examined in \cite{Lehmann} (tableaux) and \cite{PavlovicGratzl23} (generalised sequent calculus). On top of the latter we recently provided a cut-free formalisation of the basic free neutral theory of definite description \cite{IndPet}. The extension of this results to other non-trivial theories requires further study.

\end{document}